\documentclass{article}
\usepackage[T1]{fontenc}
\usepackage[utf8]{inputenc} 
\usepackage{graphicx}
\usepackage{framed} 
\usepackage{enumerate}
\usepackage{tabularx}
\usepackage{xcolor}
\usepackage{mathtools}
\usepackage{hyperref}
\usepackage{cleveref}

\usepackage{amsfonts}
\usepackage{amsthm}
\usepackage{float} 
\usepackage{dsfont}
\usepackage{amssymb}
\usepackage{enumerate}
\usepackage{enumitem}
\usepackage{breakcites}
\usepackage[ruled]{algorithm2e}

\usepackage{color}

\usepackage{geometry}
\newtheorem{definition}{Definition}[section]
\newtheorem{theorem}{Theorem}[section]
\newtheorem{lemma}{Lemma}[section]
\newtheorem{corollary}{Corollary}[section]

\newtheorem{proposition}{Proposition}[section]

\newtheorem{example}{Example}[section]



\allowdisplaybreaks

\newcommand{\Ex}[1]{\mathbb{E}\left[ #1 \right]}

\newcommand{\pr}[1]{\Pr\left[#1\right]}

\newcommand{\Adv}{\mathcal{A}}

\newcommand{\Com}[1]{\mathrm{Commit}(#1)}

\newcommand{\abs}[1]{{\left| #1 \right|}}
\newcommand{\negl}[1]{{\texttt{negl}\left( #1 \right)}} 
\newcommand{\poly}[1]{\texttt{poly}\left( #1 \right)}

\newcommand{\E}[2]{\mathbb{E}_{#1}\left[#2\right]}

\newcommand{\cN}{\mathcal{N}}
\newcommand{\cA}{\mathcal{A}}
\newcommand{\cF}{\mathcal{F}}

\newcommand{\cG}{\mathcal{G}}

\newcommand{\bN}{\mathbb{N}}
\newcommand{\bR}{\mathbb{R}}

\newcommand{\D}{\mathcal{D}}

\newcommand{\F}{\mathcal{F}}
\newcommand{\Real}{\mathsf{R}} 

\newcommand{\game}{\mathcal{G}}

\newcommand{\rSum}[3]{\sum\limits_{#1 = #2}^{#3}}

\newcommand{\PRP}{\mathsf{PRP}}

\newcommand{\Comp}{\mathsf{Comp}}

\DeclareMathOperator*{\argmax}{arg\,max}

\newcommand{\bits}{
\{0,1\}
}

\newcommand{\rX}{X}
\newcommand{\rY}{Y}
\newcommand{\rZ}{Z}
\newcommand{\rhX}{\hat{X}}
\newcommand{\rhY}{\hat{Y}}

\newcommand{\ch}{{\hat{c}}}
\newcommand{\hc}{{\hat{c}}}

\newcommand{\mEmpDom}{\underset{m, \delta}{\gtrsim}}

\newcommand{\compEmpDom}{{\gtrsim}}
\newcommand{\statDom}{\gtrsim}

\newcommand{\secparam}{\kappa}
\newcommand{\samples}{m}
\newcommand{\idealGame}{\mathcal{G}^{\mathcal{F}}}
\newcommand{\realGame}{\mathcal{G}^{\Pi}}

\newcommand{\ensemble}[1]{\{{#1}_{\secparam}\}_{\secparam\in\mathbb{N}}}

\newcommand{\ensembleU}[2]{\{{#1}^{\secparam}\left({#2}\right)\}_{\secparam\in\mathbb{N}}}

\newcommand{\hs}{\hat{s}}

\newcommand{\bids}{\mathbf{b}}

\newcommand{\values}{\mathbf{v}}

\newcommand{\meanDominance}[3]{\pr{ \frac{1}{#3}\rSum{j}{1}{#3
}{{#1}^{(j)}} > \frac{1}{#3} \rSum{j}{1}{#3}{{#2}^{(j)}} }}

\newcommand{\conMeanDominance}[4]{\pr{ \frac{1}{#3}\rSum{j}{1}{#3
}{{#1}^{(j)}} > \frac{1}{#3} \rSum{j}{1}{#3}{{#2}^{(j)}} \mid #4}}

\newcommand{\meanDominanceEQ}[3]{\pr{\frac{1}{#3} \rSum{j}{1}{#3
}{{#1}^{(j)}} \ge \frac{1}{#3} \rSum{j}{1}{#3}{{#2}^{(j)}} }}

\newcommand{\meanDominanceEqual}[3]{\pr{\frac{1}{#3} \rSum{j}{1}{#3
}{{#1}^{(j)}} = \frac{1}{#3} \rSum{j}{1}{#3}{{#2}^{(j)}} }}

\newcommand{\meanDominanceU}[3]{\pr{ \frac{1}{#3}\rSum{j}{1}{#3
}{{U_i}^{(j)}(#1)} > \frac{1}{#3} \rSum{j}{1}{#3}{U_i^{(j)}(#2)} }}

\newcommand{\Sim}{\mathsf{Sim}}

\newcommand{\vitm}{\ensuremath{M}\xspace} 
\newcommand{\vfunc}{\ensuremath{\mathcal{F}}\xspace} 
\newcommand{\vprot}{\ensuremath{\Pi}\xspace}

\begin{document}
\title{Pseudo-Equilibria, or: How to Stop Worrying About Crypto and Just Analyze the Game}
%
%

\author{
Alexandros Psomas\thanks{Purdue University. Email: \texttt{apsomas@purdue.edu}}
\and
Athina Terzoglou\thanks{Purdue University. Email: \texttt{aterzogl@purdue.edu}}
\and
Yu Wei\thanks{Georgia Tech. Email: \texttt{ywei368@gatech.edu}}
\and
Vassilis Zikas\thanks{Georgia Tech. Email: \texttt{vzikas@gatech.edu}}
}
\date{}
%
%
%
\maketitle              
\begin{abstract}
We revisit the problem of a game theorist analyzing a game that uses cryptographic protocols. Ideally, the game theorist should be able to ignore all implementation details of the cryptographic protocols and abstract them as ideal, implementation-independent primitives, in a way that conclusions in the ``ideal world'' can be faithfully transferred to the ``real world,'' where real protocols are implemented by cryptography.
Achieving this goal is crucial, as the game theorist cannot --- and should not be expected to --- grapple with the full complexity of cryptographic implementations. This is particularly relevant in the era of Web3, where the widespread adoption of distributed ledgers has created a pressing need for a common language that bridges cryptography and game theory.

We propose a new solution concept: the \emph{pseudo-Nash equilibrium}. Informally, a (poly-time) strategy profile is a pseudo-Nash equilibrium if no (poly-time) player observes a noticeable, i.e., non-negligible, (expected) utility gain by playing a different (poly-time) strategy. Pseudo-Nash is substantially simpler and more accessible to game theorists than any existing notion that attempted to address the mismatch in the (asymptotic) cryptographic method and game theory. We prove, in a very general sense, that Nash equilibria in games that use idealized, unbreakable cryptography correspond naturally to pseudo-Nash equilibria when idealized cryptography is instantiated with actual protocols (under state-of-the-world assumptions). Our translation is not only conceptually simpler than existing approaches, but also more general: it does not require tuning or restricting utility functions in the game with idealized cryptography to accommodate idiosyncrasies of cryptographic implementations. In other words, pseudo-Nash equilibria allow us to separately and seamlessly study game-theoretic and cryptographic aspects.
\end{abstract}

\section{Introduction}

(Distributed) cryptography and game theory both study interactions among agents. Cryptography traditionally takes an adversarial approach; protection is demanded against the worst-case (mis)behavior of the agents.
Game theory takes the view that agents are selfish and aim to optimize their own utility. The widespread adoption of blockchain protocols has set a stage where both viewpoints are invaluable, as decentralized systems must be both strategically robust and cryptographically secure.

However, while the interaction of the two fields has been fruitful, we believe it is fair to say that a universally adopted model that allows for the two fields to interact seamlessly has yet to emerge. Arguably, the most ambitious goal would be a model such that game theorists can work in an ``ideal world'' that completely ignores all implementation details of cryptographic protocols (i.e., treats cryptographic primitives as black boxes that provide a certain functionality), but such that game-theoretic solution concepts, e.g., equilibria, naturally carry over to the ``real-world,'' where cryptographic primitives are replaced by their (imperfect) implementations. 
Such a model, by design, would also allow cryptographers to implement cryptographic primitives without worrying about downstream effects on game-theoretic predictions.
Simply put, our work attempts to answer the following question: \emph{Is this ambitious goal attainable?}

\subsection{Our contributions}

Unfortunately, the answer to the aforementioned question is, in general, ``no.'' This is because cryptographic protocols provide guarantees only under assumptions. There are, broadly speaking, two types of assumptions cryptographic protocols require. The first type is \emph{state-of-the-world} assumptions; examples of this type are computational hardness, ideal obfuscation,  secure hardware, and access to certain ideal functionalities or setups, which are common assumptions in the design of cryptographic protocols like encryption, digital signatures, commitments, MPC, etc. The second type is \emph{behavioral} assumptions (e.g., ``honest majority''), which are assumptions about the parties that implement the protocol at hand. The reason we cannot have a model that seamlessly abstracts away all cryptographic details, but is faithful to game-theoretic solution concepts, is because of this second type. Many central solution concepts in game theory, e.g., Nash equilibria, are predictions about players' behavior. Therefore, intuitively, it is not possible to ignore elements of a game that only work under specific player behavior. In \Cref{apx: examples} (\Cref{example: nothing works for behavioral}), we present a concrete example, where the behavioral assumption of `honest majority'' fails to ensure stability against rational agents (in the context of a blockchain game).


We note that such an example is neither meant to exclude the possibility of designing a protocol that achieves a desired specification under assumptions about the utility of rational agents, nor is it intended as a claimed contribution of our work. Rather, its goal is to provide an abstract perspective on the challenges of combining cryptography and game theory, tailored to the mindset of a game theorist who wishes to use cryptography in their solution.  
Indeed, the study of rational cryptography has yielded several novel ways in which specific tasks, like secret sharing or function evaluation, can be achieved under specific rationality assumptions, e.g., parties/agents being selfish but curious~\cite{fuchsbauer2010efficient,gordon2006rational,halpern2004rational,izmalkov2005rational,kol2008games,lepinksi2005collusion}
(preferring to learn the secret while preferring that others do not). What the above example reaffirms is that although a cryptographic proof considers worst-case adversaries, i.e., makes all attacks by such adversaries ineffective, a cryptographic security (proof) of a protocol under behavioral assumptions (like honest majority) does not render it a stable solution against arbitrary rational agents. In fact, even in the context of the above example, numerous works~\cite{bahrani2024undetectable,brown2019formal,carlsten2016instability,eyal2018majority,ferreira2024computing,ferreira2022optimal,ferreira2021proof,fiat2019energy,goren2019mind,kiayias2016blockchain,sapirshtein2017optimal,tsabary2018gap} show how, under several different classes of blockchain protocols, there are several ways that miners can alter their behavior in the protocol to manipulate timestamps, their rewards in the underlying cryptocurrency, or even the value of the cryptocurrency itself. Thus, the existence of these counterexamples is well-known (but perhaps not articulated using the ideal vs real-world perspective we take in this work).

Since our goal is unattainable in general, the next natural question is whether it is attainable under conditions. And, given our discussion so far, a natural condition would be to use cryptographic protocols under state-of-the-world assumptions. Our next observation is that standard game-theoretic solution concepts do not quite work. In fact, similar counterexamples have been motivating the drive to revised notions of equilibria~\cite{DBLP:conf/focs/GarayKMTZ13,gradwohl2013sequential,Jet:HalpernP15,halpern2016sequential,DBLP:conf/wine/HalpernPS14,halpern2016computational} (see Section~\ref{subsec: related} for details). In \Cref{sec: guessing game} (\Cref{example: Nash fails}), we demonstrate a simple ``Guessing game'' to illustrate how negligible-probability deviations (i.e., brute-forcing a cryptographic commitment) can alter Nash and $\epsilon$-Nash. This occurs because attempting to break crypto on the side is weakly preferred to not doing so since it offers a potential utility gain, however small. This phenomenon can be amplified when exploiting outcomes with very large utility that occur with negligible probability.
Even though such examples rule out the Nash equilibrium (and $\epsilon$-Nash) as a general ``plug-and-play'' solution concept, it seems intuitively possible to fix the issue by slightly changing the rules. So, for example, we could ``assume-away'' the undesirable actions (e.g., players' attempts to break cryptography). However, assuming away cryptographic attacks is problematic: what does it formally mean to assume away that a player ``attempts to guess a secret value''? 

Perhaps more troubling, transitioning from a world of ideal cryptographic implementations (the ``ideal world'') to a world with real, imperfect implementations (the ``real world'') can result in two games with qualitatively different equilibria and approximate equilibria. In \Cref{apx: auction} (\Cref{example: auction where ideal vs real is weird}), we provide such an example in the context of second-price auctions. In such cases, the real-world equilibrium is not just the ideal strategy with a side attempt to ``break crypto'', but rather the incentive to exploit cryptographic implementations derails the intended mechanism design. We demonstrate that while bidding truthfully is a Nash equilibrium in the ideal world, the real-world Nash equilibrium shifts to non-truthful bidding to exploit on negligible events. 

All our examples so far paint a bleak picture, at least for Nash and $\epsilon$-Nash equilibria. Yet, a Nash equilibrium in an ideal game (where, e.g., agents do not attempt to break crypto, because it is impossible to do so) seems like the ``correct'' prediction of how agents will behave in the corresponding real games. But, if these strategies are not Nash or $\epsilon$-Nash equilibria in the real game, then what are they? In a nutshell, a key contribution of our paper is to propose, analyze, and showcase a new solution concept --- the \emph{pseudo-Nash equilibrium} --- that allows us to reconcile this gap.\footnote{We note that the term Pseudo-Nash Equilibrium was previously defined in the context of sequential games~\cite{chou1985finitely},  with very low adoption, in a definition irrelevant to indistinguishability. Given the wide adoption of the term pseudo-random for distributions indistinguishable from random, and the intuition 
of our notion, which is aimed at capturing strategies indistinguishable from Nash, we find the name Pseudo-Nash appropriate for it.
}

\textbf{Pseudo-equilibria.}
Let us recall the classical notion of Nash equilibrium (NE): A strategy profile $s = (s_1, s_2, \dots, s_n)$ is a Nash equilibrium if, for every agent $i$, the expected utility of $i$ when playing strategy $s_i$ is at least her utility when deviating to strategy $s'_i$ (when all other agents' strategies remain fixed). The above definition makes no assumption about the efficiency of strategies in $s$. However, in games that use cryptography, it is important to exclude strategies that require inefficient (e.g., super-polynomial in the cryptographic security parameter) computation. Indeed, allowing players this amount of computation makes breaking cryptography a feasible strategy for them. 
A strawman approach to solve the above conundrum is to require strategies to be polynomial-time computable. It is, however, not hard to verify that such a simple patch cannot be the solution: First, by doing so, one necessarily changes the description of the game, resulting in restricted and often complicated stability definitions that require hard-wiring cryptographic security parameters into the game---or, transitioning from simple standard games to game-families parameterized by the security parameter. Such a choice might seem natural to a cryptographer, but it substantially dilutes our goal of delivering a stability notion that is friendly to economists and game theorists with little exposure to cryptographic definitions. 

Second, and more importantly, the above modification does not even solve the problem. To see why, one can look at the Guessing Game (\Cref{example: Nash fails}): Even for a polynomial bounded player, brute-forcing the cryptographic primitive is still the best response. In fact, on a more technical note, one would need to actually bound the player's reasoning (or rationality) to be polynomial time. To do so, one would need to consider such reasoning as part of the player's strategy, which is known to lead to cumbersome definitions.  

Another strawman approach, which takes care of the above issues, would be to define pseudo-equilibrium as a strategy whose utility profile is (computationally) indistinguishable from a real equilibrium. This, however, is also quite problematic: first, in classical equilibrium definitions, one is interested in expected utilities. However, two random variables that are computationally indistinguishable may not have the same expectation. For example, consider a game where the utility random variables induced by two strategies are $U_1 = \{0\text{ w.p. } 1-\epsilon \;\&\; 1 \text{ w.p. }\epsilon\}$  and $U_2 = \{0\text{ w.p. } 1-\epsilon \;\&\; -1 \text{ w.p. }\epsilon\}$. These two random variables are computationally indistinguishable for negligible $\epsilon$, but the expected value of $U_1$ is positive, whereas that of $U_2$ is negative. Furthermore, two random variables with the same exception but completely different support (e.g., $U_1 = \{1\text{ w.p. } 1/2 \;\&\; -1 \text{ w.p. }1/2\}$  and $U_2 = \{10\text{ w.p. } 1/2 \;\&\; -10 \text{ w.p. }1/2\}$). These are equivalent for a game theorist, but are clearly distinguishable.

Finally, defining pseudo-equilibrium, e.g., pseudo-Nash, as being indistinguishable from Nash, would restrict this notion to games where their idealized-crypto variant can be shown to have a Nash equilibrium. This approach, which is implicit in~\cite{halpern2019sequential}, substantially reduces the applicability of pseudo-equilibrium and makes the notion applicable only with respect to such games, which again makes it less accessible to non-crypto-savvy domain experts.    

In this work, we take a different approach, which makes the fact that pseudo-Nash is indistinguishable from Nash (when such Nash exists) a consequence of the definition, rather than the definition itself. First, we observe that, intuitively, a necessary condition for ``strange'' strategies not to emerge as profitable deviations in any game (e.g., the ``real game'') is that events that occur with negligible probability (e.g., ``break crypto'') should be ignored. In other words, the condition ``deviations should not improve expected utility'' in the definition of Nash should be replaced; we should instead focus on different (statistical) properties of the utilities of strategies. Specifically, we build on the intuition that a player should be indifferent between two computationally indistinguishable utilities, since no PPT observer can reliably tell them apart.
Hence, a deviation whose benefit arises only on negligible-probability events should not affect strategic preferences. Our definition of a pseudo-equilibrium formalizes this intuition.

We begin by attempting to replace ``larger expectation'' in the definition of Nash with the concept of indistinguishability. Recall that two random variables $X$ and $Y$ are $\epsilon$-indistinguishable, if, for any distinguisher $\D$, $| \Pr[\D(X) = 1] - \Pr[\D(Y) = 1] | \leq \epsilon$. A first observation is that in the context of a game, where the random variables $X$ and $Y$ naturally correspond to utilities of strategies (e.g. a strategy and a deviation), this definition overlooks a crucial point: Nash equilibrium relies on comparing (the utility of) different strategies to see if one is better than (i.e., dominates) the other. (In)distinguishability, however, does not allow us to take such a decision: if a strategy $s$ is a lot better than a strategy $s'$, we should be able to distinguish them, but the existence of such a distinguisher does not help us decide which strategy we should render dominant.


A straightforward remedy is to drop the absolute value from the previous definition. Adjusting to our game-theoretic goals, perhaps we should define a strategy $s$, whose utility is distributed according to a random variable $X$, to be ``$\epsilon$-preferred'' to a strategy $s'$, whose utility is distributed according to a random variable $Y$, if, for any distinguisher $\D$, $\Pr[\D(Y) = 1] - \Pr[\D(X) = 1] \leq \epsilon$. Hence, when $X$ is a lot ``better'' than $Y$, we allow for distinguishers that pick $X$ much more often. Yet, this definition also has an issue, albeit a more nuanced one.

Concretely, requiring the above  {\em for all} distinguishers implies requiring it for the best one, i.e., the one that maximizes $\Pr[\D(Y) = 1] - \Pr[\D(X) = 1]$. Intuitively, $\D$ can be seen as implementing the best statistical test.  However, the outcome of this best distinguisher might have a weak (or no) dependence on the expectations of $X$ and $Y$.  This makes it challenging (and in some cases, impossible) to connect a notion based on this distinguisher to standard game theory concepts that only look at the expectation of (the outcome) of strategies. Thus, the corresponding equilibrium notion would not necessarily be implied by Nash equilibrium in standard games, diverging from our goal to find a version of Nash (satisfying the same intuitive stability properties) that is cryptography-friendly.

Nonetheless, the above intuition is on the right track. The key is to identify a statistical test (i.e., distinguisher) whose decision approximately follows the relation of the expectations, i.e., the test should favor strategy $X$ when the expectation of $X$ is greater than that of $Y$. This brings us to our (informal) definition of a pseudo-equilibrium, where instead of considering arbitrary distinguishers, we focus on a specific distinguisher that compares the sample means of the two random variables. As we will show, this distinguisher tells us all we need to decide which strategy should be preferred. The following (informal) definition uses the following standard notation from game theory: for any strategy profile $s = (s_1, s_2, \dots, s_n)$ (i.e., where each player $i$ plays $s_i$) we denote by $(s; s_{-i})$ the strategy profile in which player~$i$ plays $s$, and every player~$j\neq i$ plays $s_j$.

\begin{definition}[Pseudo-Nash (informal)]\label{dfn: informal pseudo}
Let $s = (s_1, s_2, \dots, s_n)$ be a strategy profile, and let $U_i = U_i(s; s_{-i})$ be the random variable that indicates the utility of player $i$ under strategy $s$, when all other players' strategies are given by $s_{-i}$. Then, $s$ is an $(m,\epsilon)$-pseudo-Nash equilibrium, if for all players $i$, and deviations $s'$ with utility random variable $U'_i =U'_i(s';s_{-i})$,
\[
\Pr\left[ \frac{1}{m} \sum_{j=1}^m U'^{(j)}_i  \geq \frac{1}{m}\sum_{j=1}^m U^{(j)}_i \right] - \Pr\left[ \frac{1}{m}\sum_{j=1}^m U^{(j)}_i \geq \frac{1}{m} \sum_{j=1}^m U'^{(j)}_i \right] \leq \epsilon,
\]
where $U^{(j)}_i$ and $U'^{(j)}_i$ are i.i.d. samples from $U_i, U'_i$ respectively.
\end{definition}

The connection between the above definition and our discussion of distinguishers/tests is evident: It can be thought of as a distinguisher who observes the utility from $m$ independent executions of the game under strategy profiles $(s;s_{-i})$ and $(s';s_{-i})$, and compares their empirical means to make its decision of dominance. This also hints towards the fact (which we prove) that the definition of pseudo-Nash equilibrium is robust against negligible-probability events. 
Indeed, rather than reasoning about the expected utility of each strategy (as in Nash equilibrium), players reason based on the empirical mean of $m$ samples; as long as $m$ is polynomial, negligible events are unlikely to be observed. This allows for the invariance of negligible events in strategic behavior. Note that the notions of negligible and polynomial are only meaningful in parameterized games. The formal definition of pseudo-Nash (Definition~\ref{def: computationalPseudo}) captures how the notion seamlessly applies to such games (as well as all classical normal form games from the game theory literature). In a nutshell, the definition of pseudo-Nash (Definition~\ref{def: computationalPseudo}) is the first notion that achieves all the following desired properties, simultaneously:




\begin{enumerate}[leftmargin=*]
\item \emph{Simplicity}: Our notion does not require any knowledge from cryptography, and can be defined for any game. This is in contrast to state-of-the-art proposals for crypto-friendly equilibria, e.g.,~\cite{halpern2016computational} which require special and complicated classes of games. See~\Cref{subsec: related} and Appendix~\ref{appendix:computationalNash} for more details.
\item \emph{Compatibility with Game Theory}: In standard (non-parameterized) games, pseudo-Nash equilibrium is well-defined and is equivalent to Nash equilibrium.
\item \emph{Indistinguishability/Compatibility with cryptography}: The replacement of random variables with computationally indistinguishable ones does not affect our notion of dominance. 
\item \emph{Unrestricted utilities}: Unlike many existing attempts to devise a notion of equilibrium for computational games which require restrictions on the utility, (e.g.,~\cite{halpern2016computational} assumes utilities of low probability events which are polynomial in a parameter of the game) a pseudo-equilibrium is not sensitive to the magnitude of the utility of low probability events. 
\end{enumerate}

In the technical part of this paper, we prove the above properties and present examples from the literature along with simple games, which showcase our notion and highlight its benefits over prior proposals. Most importantly, for the natural class of normal-form games that corresponds to games using cryptography---{\em computational games}---we can prove a general ideal-crypto-replacement theorem: 

\begin{theorem}[Informal]
Consider any (ideal) computation game $\mathcal{G^I}$ using ideal cryptographic primitives. Consider the (real) computation game $\mathcal{G^R}$  where the ideal primitive is replaced by its cryptographic implementation---secure under a state-of-the-world assumption.  If a strategy $s^\mathcal{I}$ is a Nash or pseudo-Nash equilibrium in $\mathcal{G^I}$, then the corresponding strategy $s^\mathcal{R}$---where ideal cryptography is replaced with real cryptography---is a pseudo-Nash equilibrium in  $\mathcal{G^R}$.
\end{theorem}

We note that, in addition to having all the above desirable properties, pseudo-Nash also has a rather intuitive interpretation, which reinforces our claim that it is the natural notion for the games we aim to capture, and beyond. Recall that we are interested in equilibrium for players whose reasoning is computationally limited. The definition of pseudo-Nash captures the perception---from the perspective of a player who can only think of at most $m$ iterations of the game---of whether deviating from $s_i$ dominates $s_i$ (within some error $\epsilon$.) We thus conjecture that our notion might have connections to prospect theory~\cite{eec14168-5714-3ca8-b073-d038266f2734,RePEc:kap:jrisku:v:5:y:1992:i:4:p:297-323}; investigating such connections is, in our opinion, an outstanding research direction.

\subsection{Related Work}\label{subsec: related}

Several works have attempted to use game theory within cryptographic protocols (typically MPC or its special cases, e.g., \cite{dodis2000cryptographic,fuchsbauer2010efficient,gordon2006rational,halpern2004rational,izmalkov2005rational,kol2008games,lepinksi2005collusion}).  The idea of such {\em game-theoretic cryptography} models is to incorporate incentives into the
parties' misbehavior by treating them as rational agents acting
according to some (partially) known preferences. A cryptographic protocol
execution induces a sequential game among these agents/parties and is considered ``secure'' if they have no incentive to deviate, i.e., if it induces a game-theoretic equilibrium. Notably, several of these works have highlighted mismatches between the tools used in the two disciplines. For example, \cite{kol2008games} shows that under standard equilibrium definitions, cryptography cannot be used in certain games in which privacy is important for the utility, e.g., games where selfish parties want to learn a common function on their inputs but prefer to be the only one who learns it. Most interestingly, a frequent reason for impossibilities has been rooted in the idea that it might be a best response for the adversary to try to break cryptography. Similarly, several works have used cryptographic primitives to improve game theory and mechanism design solutions (e.g.,~\cite{bradford2008protocol,essaidi2022credible,ferreira2020credible,nurmi1993cryptographic,tyagi2023riggs}), e.g., by relying on commitments to force consistency between different stages of the game.

Most relevant to our goals, several works have attempted to resolve mismatches between cryptography and game theory~\cite{DBLP:conf/focs/GarayKMTZ13,gradwohl2013sequential,Jet:HalpernP15,halpern2016sequential,DBLP:conf/wine/HalpernPS14,halpern2016computational}.
The standard approach in this area has been to either adapt the utility to account for 
computational complexity and asymptotics, e.g., by allocating a cost on computational steps~\cite{Jet:HalpernP15}, or  
develop new definitions of stability grounded in cryptography. Such definitions typically change the nature of the game, e.g., by transitioning from a simple, single game to a parameterized sequence of (usually more complex) games, in order to accommodate asymptotic reasoning, and define an equilibrium notion on top of this sequence. As such, they have typically been much more cumbersome compared to their (arguably simple and intuitive) game-theoretic analogues, such as Nash or subgame-perfect equilibria. We conjecture that this might be a reason why these new definitions, despite their evident novelty---and even support for replacing ideal cryptography with its cryptographic  implementation~\cite{DBLP:conf/focs/GarayKMTZ13,Jet:HalpernP15,DBLP:conf/wine/HalpernPS14}---have not been widely adopted by economists or game theorists, who are not necessarily familiar with theoretical computer science and cryptography idiosyncrasies. In contrast, our goal is to give such researchers a convenient tool to incorporate cryptography in their game analysis, without worrying about the aforementioned artifacts.

In order to understand the novelty and potential of the notion of pseudo-equilibrium, it is worth contrasting it with the notion of {\em computational Nash} in extensive form games~\cite{halpern2016computational}, which, to the best of our knowledge, is the state-of-the-art among equilibrium definitions that allow using cryptography within game-theory arguments.\footnote{We focus here on~\cite{halpern2016computational} but our reasoning and comparison applies to most proposed notions of equilibrium for games using cryptography, as long as they do not modify the underlying utility, e.g., by assuming that computation costs~\cite{Jet:HalpernP15}.} For readers unfamiliar with this notion, we have included its definition in Appendix~\ref{appendix:computationalNash}. In a nutshell, the notion of computational Nash~\cite{halpern2016computational} defines a new class of (sequences of) games, termed {\em Computable Uniform Sequence of Games}, which embodies the idea that cryptographic protocols can be viewed as extensive-form games, where adversarial faults correspond to deviations. In such a game definition, one needs to carefully define properties of the history, utility, and players' moves to ensure they are polynomially computable---in fact, probabilistic polynomial-time, in short PPT (see Definition~\ref{def: hp computational game}). As discussed below, this makes the notion of computational Nash applicable only to games that fit such a description.  
The above complexity makes analyzing computational Nash (see Definition~\ref{def: computational Nash})  directly a particularly challenging task. In fact, this is acknowledged in~\cite{halpern2016computational}, which offers a blueprint for finding computational Nash, which, as will become apparent below, can only be used when we are analyzing a game that is derived by replacing idealized cryptography in a sequential game with real cryptography. The core idea of this methodology is to ``reverse-engineer'' the cryptographic simulation paradigm to establish a mapping of the game using cryptography to the idealized-cryptography game. The details of this mapping are not relevant to our treatment, and the actual definition and its connection to the simulation paradigm have already been established in~\cite{halpern2016computational}. However, for readers who might be questioning whether the notion of pseudo-Nash is an important step in the direction of devising tools for non-crypto-experts to reason about games involving cryptography, we find it useful to contrast our notion to the above definition (Definition~\ref{def: ideal game represented by the computational game}).

In addition to the evident difference in complexity of the definitions, there are several qualitative and quantitative arguments in favor of pseudo-NE compared to CNE: First, as discussed above, CNE can only be defined for a very restricted class of games;  and even for this class, to make the analysis tractable one would typically require that the game itself is derived from an idealized-cryptography game, and even then, the analysis would need some degree of familiarity with advanced cryptographic reasoning (e.g., the simulation paradigm). In contrast, pseudo-NE is defined for any game. For the special class of games that use cryptography, we equip it with an automatic translation theorem that ensures that one can replace idealized cryptography with real-world cryptography.  Second, the existence of a Nash equilibrium does not necessarily imply the existence of CNE in contrast to pseudo-NE (see~\Cref{sec: guessing game,section:HTGAME}). Third, CNE is sensitive to the magnitude of the utilities of different events. For example, negligible events with huge (exponential) utility make the reasoning from~\cite{halpern2016computational} inapplicable and, as we demonstrate, eliminate the possibility of CNE. This is in contrast to pseudo-NE, which, by definition, renders negligible events irrelevant. As we show in~\Cref{section:HTGAME}, this distinction makes a big difference in the analysis of classical rational-cryptography games, e.g., rational secret sharing~\cite{halpern2004rational}. We note that this third distinction may seem somewhat esoteric, as in traditional game theory, it is reasonable to assume that utilities are constant. This is, however, not the case: first, when dealing with parameterized games that use asymptotic cryptography, it is reasonable to assume that the utilities depend on the underlying parameter. In this case, it is unclear why one should limit this dependency to be polynomial. In fact, there are natural games (even without using cryptography) in which utility would be exponential in a specific parameter (e.g., the size of the player's strategies): Imagine a version of the guessing game (\Cref{example: Nash fails}) where the leader (Player 1) chooses an $\kappa$-bit string]; then the follower (Player 2) attempts to guess the leader's choice; every bit he guesses correctly doubles his reward. It is straightforward to show that the unique Nash equilibrium has both players choose uniformly at random. However, one can show that the above strategy cannot be proven to be CNE (and the aforementioned blueprint to prove CNE from~\cite{halpern2016computational} does not help here, since there is no ideal cryptography). In contrast, we can show that the Nash strategy is also pseudo-NE. In fact, as discussed in Section~\ref{section:defPNE} we view pseudo-NE as the natural extension of Nash to such parameterized games; furthermore, our ideal-to-real theorem can be used to directly prove that replacing in the pseudo-NE of this game random coins with pseudo-random coins yields another pseudo-NE.

\section{Preliminaries}\label{sec:prelims}
Our goal is to develop tools that can be used by game-theorists and cryptographers alike,  with little to no exposure to the other field's reasoning and advanced tools. Our results offer a seamless translation from games with idealized cryptography to games with real cryptography. Given these goals, in Sections~\ref{sec:prelims_crypto}  (resp. Section~\ref{sec:prelims_gt},) below, we outline the basic definitions and tools from the cryptographic (resp. game theory) literature that a game theorist (resp. cryptographer) needs to understand the translation. Such an exposition is useful for the paper to serve its goal, but a CRYPTO reviewer is welcome to skip Sections~\ref{sec:prelims_crypto}.

\subsection{Cryptography notation and definitions}\label{sec:prelims_crypto}

State-of-the-world assumptions that are typical in cryptography rely on asymmetries that (are assumed to) exist in certain problems. E.g., one-way functions are, intuitively, easy to compute and hard to invert, pseudo-random generators are efficiently computable length-expanding deterministic algorithms but their output on random inputs is hard to distinguish from samples from a random distribution, etc. However, ``hardness'' of a problem does not make it impossible to solve; in fact, most cryptographic hardness assumptions only apply to hardness on average (i.e., on random inputs). As such, a cryptographic statement would typically embrace the idea that the adversary might with ``tiny'' probability violate its claimed security (by violating the underlying assumption).  

Capturing ``tiny'' in a way that allows for cryptographic proofs is tricky. For example, in an encryption scheme, we would want that the adversary has a tiny probability of recovering information about the plaintext. But if the key is small, say 5 bits, then with probability $2^{-5}=1/32$ the adversary can guess it and recover the whole plaintext. As such, the notion of tiny is defined to correspond to ``eventually tiny,'' i.e., for cryptographic keys whose size is larger than some value $\secparam_0$, usually referred to as the {\em security parameter}. This gives rise to the following notion of {\em negligibility} aimed at capturing the above intuition of eventually tiny:

\begin{definition}\label{def: negl}
    A function $\delta:\bN \rightarrow \bR$ is negligible if for every constant $c$  there exists a $\secparam_c\in\mathbb{N}$ such that for all $\secparam
    \ge \secparam_c$ it holds $\delta(\secparam)<\frac{1}{\secparam^c}$.
\end{definition}

Most cryptographic security definitions compare a cryptographic construction, that typically involves a key, to an ideal primitive, that the construction is supposed to securely realize. Security requires that the random variable that corresponds to the outcome of an execution of the cryptographic construction cannot be distinguished from the random variable that corresponds to the outcome of an invocation of the ideal primitive (we defer details of such a definition to Section~\ref{sec: pseudo-for-crypto-games}). As such, the notion of (computational) indistinguishability is deeply rooted in the cryptographic method (and as we shall see, will be the key in defining a crypto-friendly notion of stability in games that use cryptography). We shortly define computational indistinguishability. 

For a random variable $\rX$ and a randomized algorithm $\D$ with a binary output, that can draw a sample from the probability distribution $P_{\rX}$ of $\rX$, we denote by $\D^{\rX}(x)$ the random variable which corresponds to the outcome of $\D$ on input an $x$ sampled from $\rX$. We also denote by $1^\secparam$ the unary representation of $\secparam$, i.e., the string consisting of $\secparam$ ones.

\begin{definition}[Computational Indistinguishability]\label{def: compInd}
    Let $\rX=\{\rX_\secparam\}_{\secparam\in\mathbb{N}},$ $\rY=\{\rY_\secparam\}_{\secparam\in\mathbb{N}}$ be a pair of ensembles of random variables.  $\rX$ is computationally indistinguishable from $\rY$, denoted by $\rX \cong \rY$, if there exists a negligible function $\negl{\cdot}:\bN \rightarrow [0,1]$
    such that for every probabilistic polynomial time (PPT) algorithm $\D$ with a binary output, called the {\em distinguisher}, the following holds: 
    \begin{align*}
\abs{\pr{\D(\rX_\secparam;1^\secparam) = 1} - \pr{\D(\rY_\secparam; 1^\secparam) = 1}} \leq \negl{\secparam}.
    \end{align*}
\end{definition}

We will also use the standard closure of computational indistinguishability under taking polynomially many independent samples~\cite{GoldreichBook1}. Namely, if $X \cong Y$, then for every constant $c\in\mathbb{N}$, the joint distributions of $\secparam^c$ i.i.d.\ samples from $X_\secparam$ and from $Y_\secparam$ are computationally indistinguishable (by a hybrid argument).

The quantity $\abs{\pr{\D(\rX_\secparam;1^\secparam) = 1} - \pr{\D(\rY_\secparam; 1^\secparam) = 1}}$ is often referred to as the {\em distinguishing advantage} of $\D$ in distinguishing 
$\rX_\secparam$ from $\rY_\secparam$, as it captures, intuitively, how similar $\rX_\secparam$ and $\rY_\secparam$ look to $\D$. In fact, it is not hard to see that if we do not restrict $\D$ to be PPT, i.e., if we take the distinguishing advantage of the best (even inefficient) distinguisher, then this will be equal to the standard {\em statistical distance} between the  $\rX_\secparam$ and $\rY_\secparam$. For the curious reader, we note in passing that providing $\D$ the input $1^\secparam$ is to ensure that it runs in time polynomial in $\secparam$, which is needed to avoid trivial counter-examples of rendering cryptographic constructions secure because the distinguisher does not have sufficient time to parse $\rX_\secparam$ (see~\cite{GoldreichBook1} for a detailed discussion).

\subsection{Game theory notations and definitions}\label{sec:prelims_gt}

We define the notion of pseudo-equilibrium on standard normal form games. We denote by $\game([n],\cN, U)$ a game with $n$ players and with action set $\cN=\cN_1\times\ldots\times\cN_n$. For a mixed strategy profile $s\in \cN$ we are interested in the utility random variable $U_i(s_i,s_{-i})$ of player $i\in[n]$, where $s_i$ denotes the strategy of player $i$ and $s_{-i}$ the strategies of all other players. In such games, the Nash equilibrium is defined as usual:

\begin{definition}[Nash Equilibrium]
A strategy profile $s=(s_1,\ldots,s_n)$ is a Nash equilibrium (NE) of a game $\game$ if, for all $i\in[n]$ and for all deviating strategies $\hs_i$ it holds $\E{}{U_i(s)}\ge\E{}{U_i(\hs_i,s_{-i})}.$ 
\end{definition}

A common relaxation of NE to an approximate notion, called $\epsilon$-Nash, allows for strategies that are suboptimal at most $\epsilon$ (additive) factor. 

\begin{definition}[$\epsilon$-Nash Equilibrium]
A strategy profile $s=(s_1,\ldots,s_n)$ is an $\epsilon$-Nash equilibrium ($\epsilon$-NE) of a game $\game$ for some $\epsilon\geq 0$ if, for all $i\in[n]$ and for all deviating strategies $\hs_i$ it holds $\E{}{U_i(s)}\ge\E{}{U_i(\hs_i,s_{-i})}-\epsilon.$ 
\end{definition}

Looking ahead in the technical part of our paper, we remark that our notion of pseudo-equilibrium will allow us to seamlessly extend the above classical notions also to parameterized games---like the guessing game (\Cref{example: Nash fails}), modified so that player one has a variable-size input (or input that depends on a cryptographic security parameter) and the utilities depend on this size. 

\section{Empirical Dominance and Pseudo-Nash Equilibrium}\label{section:defPNE}

We introduce a new notion of ordering random variables with respect to their empirical means. We say that a random variable  $\rX$ $(m,\delta)$-empirically means dominates $\rY$ if the empirical mean of $\rX$ tends to be higher than $\rY$. Formally,

\begin{definition}[Empirical Mean Dominance]\label{def:emd}
    Let $\rX, \rY$ be a pair of random variables. Then $\rX$ {\bf $(\samples, \delta)$-empirical mean dominates}(EM-dominates) $\rY$, denoted by $\rX \mEmpDom \rY$, for $\samples \in \mathbb{N}^+$ and $\delta \in [0,1]$:
    \begin{align*}
        \meanDominanceEQ{\rY}{\rX}{m}-\meanDominanceEQ{\rX}{\rY}{m} \le \delta\\
    \end{align*}
    where $\rX^{(j)},\rY^{(j)}$ are i.i.d. samples from $\rX,\rY$ respectively.
\end{definition}

As the name suggests, our goal is to use empirical dominance to compare (mixed) strategies, $s$ and $\hat{s}$, in normal-form games. In this context, the random variables $X$ and $Y$ will indicate the utilities of $s$ and $\hat{s}$, respectively in the game.\footnote{Note that we mean the actual utility that players receive when they play these strategies rather than expected utility.}
In order to account for computational considerations and enable cryptographic reasoning, we make the following restriction on empirical dominance. 

\begin{definition}[Computational Mean Dominance]\label{def:computational-mean-dominance}
   Let $\rX=\{\rX_\secparam\}_{\secparam\in\mathbb{N}},$ $\rY=\{\rY_\secparam\}_{\secparam\in\mathbb{N}}$ be a pair of random variable ensembles.  Then $\rX$ {\bf computational} mean dominates $\rY$, denoted by $\rX \statDom \rY$,
   if  for all constants $c \geq 1$, there exist constants $\ch>c,\;\secparam_0>0$ such that for all $\secparam\ge\secparam_0$: 
    \begin{align*}
       \meanDominanceEQ{\rY_\secparam}{\rX_\secparam}{\secparam^{\ch}}-\meanDominanceEQ{\rX_\secparam}{\rY_\secparam}{\secparam^{\ch}} < \frac{1}{\secparam^{c}}\\
    \end{align*}
    where $\rX_{\secparam}^{(j)},\rY_{\secparam}^{(j)}$ are i.i.d. samples from $\rX_\secparam,\rY_\secparam$ respectively.
\end{definition}

As discussed in the introduction, the intuitive interpretation of the above $\rX \statDom \rY$ dominance definition is as follows. Consider a
distinguisher who uses the empirical means of $X$ and $Y$ as a statistical test to decide whether or not the expected value of $X$ is above the expected value of $Y$. As long as this distinguisher takes at most polynomially many---i.e., $\secparam^{\ch}$ for some constant $\ch$---samples of $X$ and $Y$, with all but negligible ---eventually diminishing faster than $1/k^c$ for any $c$---probability, he will decide that $X$ dominates $Y$.

In the context of analyzing games, the parameter $\secparam$ in the above ensembles corresponds to the size of the game's description, which is a natural input to any distinguisher performing the above statistical test. As such the above notion applies to comparing two strategies $s$ and $\hat{s}$ in standard (fixed) games similarly to Definition~\ref{def:emd}, by setting $X_\secparam=X_1$ and $Y_\secparam=Y_1$, for all $k\in\mathbb{N}$, where $X_1$ and $Y_1$ are the utilities of the two strategy profiles, $s$ and $\hat{s}$, respectively. In fact, somewhat surprisingly,  as we prove in Lemmas~\ref{lemma: E[X]>E[Y]impliesComDominance} and~\ref{lemma: E[X]=E[Y]impliesComDominance},  restricting the notion of empirical dominance as above yields a definition equivalent to standard dominance of expectations, something which is not true otherwise. Looking ahead this will render our pseudo-Nash notion equivalent to Nash in such games, which in our opinion is the ultimate "sanity check" for any such notion.

Most interestingly for the goals of our paper,  by incorporating ensembles into the definition, we are now able to reason about parameterized games---think of the guessing game from our introduction where the size of the input string is a parameter and affects the utility (e.g., the more bits of player 1's input that player 2 guesses the higher his utility). We believe that pseudo-Nash is the natural adaptation of classical Nash to such games. Most importantly,  it means that we can now seamlessly capture games using cryptography by simply making the security parameter (in unary notation, $1^\secparam$) part of the game description.

Given the above notion of dominance, we can now define our notion of pseudo-Nash equilibrium. Similar to Nash equilibrium, a strategy profile is a pseudo-equilibrium if for all players, the random variable (ensembles) that corresponds to their utility dominates all other utility random variables that come from a unilateral deviation. Formally, we define pseudo-equilibrium.

\begin{definition}[Pseudo-equilibrium]
    \label{def: computationalPseudo}
    A computational pseudo-equilibrium of a (possibly parameterized) game $\game([n], \cN,U)$ 
    is a strategy profile $s = (s_1, \cdots, s_n)\in\cN$ that for all $i\in [n]$ and for all $\hs_i\in\cN_i$, the utility random variable (ensemble) of player $i$ according to $s_i$, $ U_i(s_i,s_{-i}) = \ensembleU{U_i}{s} $, computationally  EM dominates the utility variable (ensemble) of player $i$ according to $\hs_i$, $U_i(\hs_i,s_{-i}) = \ensembleU{U_i}{(\hs_i,s_{-i)}}$.  That is, $s$ is a computational pseudo-equilibrium if $\;\forall i\in [n],\forall\hs_i\in\cN_i: U_i(s_i,s_{-i}) \compEmpDom U_i(\hs_i,s_{-i})$.  
\end{definition}

\subsection{Equivalence to Nash for Standard Games}

In this section, we show that for standard fixed description (i.e., non-parameterized) games, our pseudo-equilibrium notion is equivalent to Nash equilibrium.

\begin{theorem}
\label{thrm: Nash equiv Pseudo}
For a normal-form game $\game([n],\cN,U)$,  with bounded utilities, a strategy profile $s=(s_1,\ldots,s_n)$ is a Nash equilibrium if, and only if, it is a pseudo-Nash equilibrium (\Cref{def: computationalPseudo}).
\end{theorem}

The proof relies on three technical lemmas relating computational mean dominance to expected utility comparisons for bounded-support random variables. 
Full proofs appear in ~\Cref{app:equiv-nash-proofs}.

\begin{lemma}
    \label{lemma: E[X]>E[Y]impliesComDominance}
    Let $\rX,\rY$ be two random variables with bounded support $R$ such that $\E{}{\rX}>\E{}{\rY}$. Then $\rX$ computationally dominates $\rY$, $\rX\compEmpDom\rY$(\Cref{def:computational-mean-dominance}).  
\end{lemma}

\begin{lemma}
     \label{lemma: E[X]=E[Y]impliesComDominance}
     Let $\rX,\rY$ be two random variables with bounded support $R$ such that $\E{}{\rX}=\E{}{\rY}$. Then $\rX$ computationally dominates $\rY$, $\rX\compEmpDom\rY$, and $\rY$ computationally dominates $\rX$, $\rY\compEmpDom\rX$, according to \Cref{def:computational-mean-dominance}.
\end{lemma}

\begin{lemma}
     \label{lemma: XdomY implies E[X]>=E[Y]}
     Let $\rX,\rY$ be two random variables with bounded support $R$. If $\rX$ computationally dominates $\rY$, according to \Cref{def:computational-mean-dominance}, then $\Ex{\rX}\ge\Ex{\rY}$.
\end{lemma}

\begin{proof}[Proof of~\Cref{thrm: Nash equiv Pseudo}] Given the above three lemmas, the proof of~\Cref{thrm: Nash equiv Pseudo} proceeds as follows:

\noindent We will show each direction separately.

\noindent($\implies$) \textbf{Nash implies Pseudo.} We want to show that a Nash equilibrium also satisfies \Cref{def: computationalPseudo}. That is, there exists a $\secparam_0$ such that for all $\secparam\ge\secparam_0$ and for all players $i$ and for all alternative strategies $\hs_i$ : 

    \begin{align*}
    &\meanDominanceU{\hs_i;s_{-i}}{s}{\secparam^{4c}}\\ 
    &\qquad\qquad\qquad-\meanDominanceU{s}{\hs_i;s_{-i}}{\secparam^{4c}} \le \frac{1}{\secparam^c} 
    \end{align*}

Since the strategy profile $s=(s_1,\ldots,s_n)$ is a Nash equilibrium, we have for all $i\in[n]$ and for every unilateral deviation $\hat{s}_i \in \cN_i$ of player $i$ the $\E{}{U_i(s_i;s_{-i})} \ge \E{}{U_i(\hat{s}_i;s_{-i})}.$ Let us define $\Delta = \min_{i,\hs_i} \{\E{}{U_i(s_i;s_{-i})} - \E{}{U_i(\hs_i;s_{-i})} \} \ge 0$. By construction of the game, the utilities are bounded, i.e., for all strategy profiles $U_i(\cdot) \in [a,b]$. Let $R = \abs{b-a}$ be the range. Consider player $i$ and the deviating strategy $\hs_i$ that minimizes the difference in expectation $\Delta$. 

\noindent\textit{Case 1:}  $\Delta > 0$. In this case, we can directly apply~\Cref{lemma: E[X]>E[Y]impliesComDominance}, to show that  $U_i(s_i;s_{-i})\compEmpDom U_i(\hs;s_{-i}).$

\noindent\textit{Case 2:}  $\Delta = 0$ In this case we can directly apply \Cref{lemma: E[X]=E[Y]impliesComDominance}

\noindent($\impliedby$) \textbf{Pseudo implies Nash.}
Assume that the strategy profile $s = (s_1,\dots,s_n)$ is a pseudo-Nash equilibrium. By definition, for every player $i\in [n]$, for every unilateral deviation $\hat{s}_i \in \cN_i$, for every constant $c$there exists a sample size $\secparam_0$ such that for all $\secparam\ge \secparam_0$ the following holds:
    
    \begin{align*}
    &\meanDominanceU{\hs_i;s_{-i}}{s}{\secparam^{4c}}\nonumber\\
    &\qquad\qquad\qquad- \meanDominanceU{s}{\hs_i;s_{-i}}{\secparam^{4c}} \le \frac{1}{\secparam^c}  
    \end{align*}

where $U_i^{(j)}(s_i,s_{-i})$ and $U_i^{(j)}(\hat{s}_i,s_{-i})$ denotes $j^{\text{th}}$ independent sample from the corresponding utility random variables. By~\Cref{lemma: XdomY implies E[X]>=E[Y]}, this condition implies that for every $i\in [n]$ and every deviation $\hat{s}_i\in \cN_i$,
\[
\E{}{U_i(s_i;s_{-i})} \ge \E{}{U_i(\hat{s}_i;s_{-i})},
\]
which is exactly the Nash equilibrium condition.

\end{proof}

\subsection{Beyond Nash for Parameterized Games} 

Recall that our goal is to define a notion of stability for games using cryptography, which eliminates unnatural Nash equilibria (e.g., brute-forcing crypto). The above equivalence theorem makes one wonder, if Nash is equivalent to pseudo-Nash, how can our notion provide more stability if it collapses to the very concept it seeks to improve? The resolution is that equivalence between Nash and pseudo-Nash does not hold in parametrized games (i.e., games where the utility depends on the security parameter $\secparam$). As a sanity check, we describe two  ensembles $\rX=\ensemble{\rX}, \rY=\ensemble{\rY}$ for which  $\E{}{\rY_\secparam} > \E{}{\rX_\secparam} $ for all $\secparam$; but, the empirical mean of $\rY_\secparam$ does not converge to its expectation with polynomial samples in $\secparam$, which would make the two equilibrium notions distinct.

    Consider the following two random variable ensembles $\rX=\ensemble{\rX},$ and $ \rY=\ensemble{\rY}$. Let $\rX_\secparam = \{0 \;\text{w.p.}\; 1/2,2 \;\text{w.p.}\; 1/2 \}$ and $\rY_\secparam = \{0 \;\text{w.p.}\; 1-1/2^\secparam,2^{2\secparam} \;\text{w.p.}\; 1/2^\secparam\}$. It is immediate that $\E{}{\rX_\secparam}=1$ and $\E{}{\rY_\secparam}=2^{\secparam}$, hence $\E{}{\rX_\secparam} < \E{}{\rY_\secparam}$ for all $\secparam$. However, we will show that $\rX$ computationally EM dominates $\rY$, $\rX \compEmpDom \rY$. Intuitively, the probability that $\rSum{j}{1}{\secparam^c}\rY_\secparam>0$ is negligible. On the other hand, the probability that $\rSum{j}{1}{\secparam^c}\rX_\secparam=0$ is also negligible.

    To show that $\rX$ computationally dominates $\rY$ we need to show that for all $c\ge1$ there exists a $\ch>c$ and a $\secparam_0$ such that for all $\secparam\ge\secparam_0$:
    \begin{equation}\label{eq: exampleComp}
        \meanDominance{\rY_\secparam}{\rX_\secparam}{\secparam^\hc} -\meanDominance{\rX_\secparam}{\rY_\secparam}{\secparam^\hc}\le\frac{1}{\secparam^c} \\
    \end{equation}

Rearranging the LHS of~\Cref{eq: exampleComp}

    \begin{align*}
        &\meanDominance{\rY_\secparam}{\rX_\secparam}{\secparam^\hc} -\meanDominance{\rX_\secparam}{\rY_\secparam}{\secparam^\hc} \\
        ={}&\meanDominance{\rY_\secparam}{\rX_\secparam}{\secparam^\hc}- \left(1 -\meanDominanceEQ{\rY_\secparam}{\rX_\secparam}{\secparam^\hc}\right)\\
        ={}&2\meanDominance{\rY_\secparam}{\rX_\secparam}{\secparam^\hc}+\meanDominanceEqual{\rY_\secparam}{\rX_\secparam}{\secparam^\hc} -1\\
    \end{align*}

To show that $\rX$ dominates $\rY$ we will upper bound the probabilities separately. We will show that both probabilities are negligible and therefore the difference is negative, thus strictly less than $\frac{1}{\secparam^c}$.

First,
    \begin{align*}
    \meanDominance{\rY_\secparam}{\rX_\secparam}{\secparam^\hc} &= 1-\pr{\forall j\in[\secparam^\hc]: \rY_\secparam^{(j)}=0} \\
        &=1-\left(1-\frac{1}{2^\secparam}\right)^{\secparam^\hc}\le 1-1+\frac{\secparam^\hc}{2^\secparam}\le \negl{\secparam}
    \end{align*}
        
   Next,

   \begin{align*}
       \meanDominanceEqual{\rY_\secparam}{\rX_\secparam}{\secparam^\hc} &= \pr{\forall j\in[\secparam^\hc]: \rY_\secparam^{(j)}=0} \pr{\forall j\in[\secparam^\hc]: \rX_\secparam^{(j)}=0}\\
       &\le\pr{\forall j\in[\secparam^\hc]: \rX_\secparam^{(j)}=0}=\frac{1}{2^{\secparam^\hc}}\le \negl{\secparam}
   \end{align*}

That concludes the example that demonstrates a random variable dominates a random variable with much higher expected utility.

\section{Properties of Pseudo-Nash and Crypto-Friendliness} %

The previous section demonstrated that computational mean dominance (and pseudo-Nash) ``plays well'' with game theory. In this section, we demonstrate that they also play well with cryptography, developing the basic tools that will allow us to prove our main theorem (stability-preserving ideal to real cryptography transition). More concretely, 
we will show that the definition of computational mean dominance (\Cref{def:computational-mean-dominance}) is compatible with computational indistinguishability (\Cref{def: compInd}). The following results show that: (i) if two random variables are computationally indistinguishable then we have bidirectional dominance (\Cref{lemma: comp ind implies bidirectional dominance}) (ii) if a random variable $\rX$ dominates a random variable $\rY$ then changing either with an indistinguishable one will not affect the dominance (\Cref{prop: EM-domination and indistinguishability for X} for replacing the dominating ensemble;  ~\Cref{prop: EM-domination and indistinguishability for Y} for replacing the dominated ensemble.)
Conceptually, these results formalize the intuition that negligible-probability events should not affect incentives. In particular, \Cref{lemma: comp ind implies bidirectional dominance} shows that a player should be indifferent between two computationally indistinguishable utility ensembles, since no PPT observer given polynomially many samples can reliably tell them apart. Consequently, a deviation that improves utility only on a negligible-probability event should not affect preferences. This is captured by \Cref{prop: EM-domination and indistinguishability for X,prop: EM-domination and indistinguishability for Y} showing that computational mean dominance is invariant under replacing either the dominating or dominated ensemble by a computationally indistinguishable one, even if the replacement changes the expected utility by a large amount.


First, we show that if two ensembles of random variables $\rX = \ensemble{\rX}$, $\rY=\ensemble{\rY}$ are computationally indistinguishable, then we have bidirectional computational empirical mean dominance.  

\begin{lemma}\label{lemma: comp ind implies bidirectional dominance}
    Let  $\rX = \ensemble{\rX}$, $\rY=\ensemble{\rY}$ be a pair of random variable ensembles such that $\rX$ and $\rY$ are computationally indistinguishable. Then $\rX$ computationally-EM dominates $\rY$,  and $\rY$ computationally-EM dominates $\rX$. That is $\rX\cong\rY \implies \rX\compEmpDom\rY\;\&\;\rY\compEmpDom\rX$.  
\end{lemma}

\begin{proof}
    \label{proof: lemma: comp ind implies bidirectional dominance}
    We prove this lemma by contradiction. Assume that for two computationally indistinguishable random variables ensembles $\rX$ and $\rY$, the bidirectional dominance does not hold. Since $\rX\compEmpDom\rY$ and $\rY\compEmpDom\rX$ are symmetric, without loss of generality, we assume that $\neg(\rX\compEmpDom\rY)$ holds.  Then we can construct a probabilistic polynomial-time (PPT) distinguisher that distinguishes $\rX$ from $\rY$ with a non-negligible advantage. 
    

     By definition, $\neg(\rX\compEmpDom\rY)$ implies that there exists a constant $c$ such that for all $\ch>c$, we set $\ch=4c$, and for all $\secparam_0$ there exists a $\secparam\ge\secparam_0$ such that:

\begin{equation}\label{eq: X not doms Y}
    \meanDominance{\rY_\secparam}{\rX_\secparam}{\secparam^{4c}} -\meanDominance{\rX_\secparam}{\rY_\secparam}{\secparam^{4c}} \ge \frac{1}{\secparam^c} 
\end{equation}


    Consistent with the cryptographic literature, we assume that the distinguisher is non-uniform (i.e., can take a polynomial advice, which in our case will be the constant $c$ guaranteed to exist from the above observation). We consider the distinguisher $\D_1: \bits^* \to \{0,1,\bot\}$ (\Cref{alg: distingusiher X and Y}), that takes as input the running time token $1^\secparam$  and $\secparam^{4c}$ i.i.d.\ samples from $\rZ_\secparam$, where $\rZ_\secparam$ is equally likely distributed according to $\rX_\secparam$ or $\rY_\secparam$. The distinguisher then outputs (i) $\bits$ that indicates whether the samples were drawn from $\rX_\secparam$ or $\rY_\secparam$ respectively (ii) $\bot$, which represents that $\D_1$ chooses to abort from making a guess.

\begin{figure}[ht]
    \centering
    \framebox{%
    \parbox{0.8\textwidth}{%
    \bigskip
    \centering
    \textbf{Distinguisher} $\D_1(1^\secparam)$
   
    \begin{enumerate}
        
        \item $\D_1$ gets $c$ as its advice. 
        \item Compute the empirical mean of $\secparam^{4c}$ i.i.d. samples from $\rZ$.
        
        Let $\overline{Z}=\frac{1}{\secparam^{4c}}\rSum{i}{1}{\secparam^{4c}}\rZ^{(i)},$ 
        
        \item Compute the empirical mean of $\secparam^{4c}$ i.i.d. samples from $\rX_\secparam$.
        
        Let $\overline{X}=\frac{1}{\secparam^{4c}}\rSum{i}{1}{\secparam^{4c}}\rX^{(i)}_\secparam,$  
        
        \item If $\overline{\rZ} > \overline{\rX} $, output $1$; If $\overline{\rZ} < \overline{\rX} $, output $0$; Otherwise output $\bot$.

    \end{enumerate}
    }
    }
    \caption{Distinguisher $\D_1$ for $\rX, \rY$}
    \label{alg: distingusiher X and Y}

    \end{figure}

    First, we analyze the performance of the distinguisher $\D_1$ when the samples are drawn from $\rX_\secparam$.
    \begin{align}
        & \pr{\D_1(\rX^{(1)}_\secparam, \cdots, \rX^{(\secparam^{4c})}_\secparam; 1^{\secparam}) = 0} - \pr{\D_1(\rX^{(1)}_\secparam, \cdots, \rX^{(\secparam^{4c})}_\secparam; 1^{\secparam}) = 1} \nonumber\\
        ={}& \pr{\overline{\rX}_\secparam' < \overline{\rX}_\secparam  } - \pr{\overline{\rX}_\secparam' > \overline{\rX}_\secparam } \nonumber \\
        ={}& 0 \label{ineq: proof: 1}.
    \end{align}

    Where $\overline{\rX}_\secparam'$ and $\overline{\rX}_\secparam$ are independent copies of the empirical mean random variable of $\rX_\secparam$ with $\secparam^{4c}$ samples. The last equality holds due to symmetry.
    
    Similarly, we obtain the following when $\D_1$ has access to $\rY$.
    \begin{align}
        & \pr{\D_1(\rY_{\secparam}^{(1)}, \cdots, \rY_{\secparam}^{(\secparam^{4c})}; 1^{\secparam}) = 1} - \pr{\D_1(\rY_{\secparam}^{(1)}, \cdots, \rY_{\secparam}^{(\secparam^{4c})}; 1^{\secparam}) = 0} \nonumber\\
        ={}&\pr{\overline{\rY}_\secparam >  \overline{\rX}_\secparam} -\pr{\overline{\rX}_\secparam > \overline{\rY}_\secparam} \nonumber\\
        ={}&\meanDominance{\rY_\secparam}{\rX_\secparam}{\secparam^{4c}} -\meanDominance{\rX_\secparam}{\rY_\secparam}{\secparam^{4c}} \nonumber\tag{From \Cref{eq: X not doms Y}} \\
        \geq{}& \frac{1}{\secparam^c}.\label{ineq: proof: 2}
    \end{align}

    Summing up Inequality~(\ref{ineq: proof: 1}) and Inequality~(\ref{ineq: proof: 2}), we get 
    \begin{align*}
        &\left(\pr{\D_1(\rX^{(1)}_\secparam, \cdots, \rX^{(\kappa^{4c})}_\secparam; 1^{\secparam}) = 0} - \pr{\D_1(\rY_{\secparam}^{(1)}, \cdots, \rY_{\secparam}^{(\secparam^{4c})}; 1^{\secparam}) = 0} \right) \\
        +{}&\left(\pr{\D_1(\rY_{\secparam}^{(1)}, \cdots, \rY_{\secparam}^{(\secparam^{4c})}; 1^{\secparam}) = 1} - \pr{\D_1(\rX^{(1)}_\secparam, \cdots, \rX^{(\secparam^{4c})}_\secparam; 1^{\secparam}) = 1} \right) \geq \frac{1}{\secparam^c},
    \end{align*}

    The final inequality implies that the distinguishing advantage of the distinguisher is at least $\frac{1}{2\secparam^c}$, which is non-negligible.

    From the existence of the distinguisher $\D_1$, we have that there exists a PPT distinguisher $\D_2$ that can distinguish $\rX_\secparam$ from $\rY_\secparam$ with non-negligible advantage, even when given only a single sample from either $\rX_\secparam$ or $\rY_\secparam$. This result follows from a standard result in the cryptographic literature~\cite{GoldreichBook1}. The existence of $\D_2$ implies that $\neg\big(\rX\cong\rY\big)$ and we reach a contradiction.
\end{proof}

\begin{proposition}
    \label{prop: EM-domination and indistinguishability for X}
    Let $\rX=\{\rX_\secparam\}_{\secparam\in\mathbb{N}}, \rhX=\{\rhX_\secparam\}_{\secparam\in\mathbb{N}}, \rY=\{\rY_\secparam\}_{\secparam\in\mathbb{N}}$ be random variable ensembles, such that $\rX,\rhX$ are computationally indistinguishable. Then $\rX$ computationally dominates $\rY$ if, and only if, $\rhX$ computationally  dominates $\rY$. That is, if $\rX \cong \rhX$ then  $\rX \compEmpDom \rY \Leftrightarrow\rhX \compEmpDom \rY$.
\end{proposition}

\begin{proof}
    \label{proof: prop: EM-domination and indistinguishability for X}
   We prove the forward direction; the reverse direction follows by swapping $\rX$ and $\rhX$. Assume $\rX \compEmpDom \rY$ and $\rX \cong \rhX$ but $\neg(\rhX \compEmpDom \rY)$. We will derive a contradiction by constructing a PPT distinguisher.

    From $\rX \compEmpDom \rY$ we have that for all constants $c_1\ge1$, there exists a constant $d_1>c_1$ and $\secparam_1$ such that for all $ \secparam\ge\secparam_1,$
    \begin{equation}\label{eq: X doms Y} \meanDominance{\rY_\secparam}{\rX_\secparam}{\secparam^{d_1}} -\meanDominance{\rX_\secparam}{\rY_\secparam}{\secparam^{d_1}} < \frac{1}{\secparam^{c_1}}
    \end{equation}

    Similarly,  from  $\neg(\rhX \compEmpDom \rY)$ we have that there exists a constant $c_2\ge1$ such that for all $d_2>c_2$ and $\secparam_2$ there exists $\secparam\ge\secparam_2$ with

    \begin{equation}\label{eq: X' not doms Y}
      \meanDominance{\rY_\secparam}{\rhX_\secparam}{\secparam^{d_2}} -\meanDominance{\rhX_\secparam}{\rY_\secparam}{\secparam^{d_2}} \ge  \frac{1}{\secparam^{c_2}}   
    \end{equation}

First, we fix the polynomial sample size used by the distinguisher. Let $c_2\ge 1$ be the constant guaranteed by \Cref{eq: X' not doms Y}, and set $c_1 := c_2+1$. Since \Cref{eq: X doms Y} holds for all constants, there exists a constant $d>c_1$ and a threshold $\secparam_1$ such that \Cref{eq: X doms Y} holds for all $\secparam\ge \secparam_1$ when using this $d$. Because $d>c_1>c_2$, we may also instantiate the quantifier over $d_2>c_2$ in \Cref{eq: X' not doms Y} with the same exponent $d$. We therefore fix this $d$ and let the distinguisher use $\secparam^d$ samples throughout.

    Next, we show that the PPT distinguisher $\D$ in \Cref{alg: distingusiher X} has a non-negligible advantage of distinguishing between $\rX$ and $\rhX$. Consistent with the cryptographic literature, we assume that the distinguisher is non-uniform (i.e., can take a polynomial advice, which in our case will be the constant $d$ guaranteed to exist from the above observation). That is, there exists a polynomial $\poly{\cdot}$ that for all $\secparam_0$ there exists a $\secparam\ge\secparam_0:$
\begin{equation*}\label{eq: distinguishing advantage X,X'}  \abs{\pr{\D(\rX_{\secparam}^{(1)},\cdots,\rX_{\secparam}^{(\secparam^d)};1^\secparam)=1}
-\pr{\D(\rhX_{\secparam}^{(1)},\cdots,\rhX_{\secparam}^{(\secparam^d)};1^\secparam)=1}}
\ge \frac{1}{\poly{\secparam}}.
    \end{equation*}

    \begin{figure}[ht]
    \centering
    \framebox{%
    \parbox{0.8\textwidth}{%
    \bigskip
    \centering
    \textbf{Distinguisher} $\D(1^\secparam)$
   
    \begin{enumerate}
        
        \item $\D$ gets $d$ as its advice. 
        \item Compute the empirical mean of $\secparam^{d}$ i.i.d. samples from $\rZ_\secparam$.
        
        Let $\overline{Z}=\frac{1}{\secparam^{d}}\rSum{i}{1}{\secparam^{d}}\rZ_\secparam^{(i)},$ 
        
        \item Compute the empirical mean of $\secparam^{d}$ i.i.d. samples from $\rY_\secparam$.
        
        Let $\overline{Y}=\frac{1}{\secparam^{d}}\rSum{i}{1}{\secparam^{d}}\rY^{(i)}_\secparam,$  
        
        \item If $\overline{\rZ} > \overline{\rY} $, output $1$; If $\overline{\rZ} < \overline{\rY} $, output $0$; Otherwise output $\bot$.

    \end{enumerate}
    }
    }
    \caption{Distinguisher $\D$ for $\rX, \rhX$}
    \label{alg: distingusiher X}
    \end{figure}

    To lower-bound the distinguishing advantage of $\D$, we must ensure that for every $\secparam_0$ there exists
$\secparam\ge \secparam_0$ for which both \eqref{eq: X doms Y} and \eqref{eq: X' not doms Y} hold (with our fixed constants).
Let $\secparam_1$ be the threshold from \eqref{eq: X doms Y}, and define $\secparam_2:=\max\{\secparam_0,\secparam_1\}$.
By \eqref{eq: X' not doms Y}, there exists $\secparam\ge \secparam_2$ such that \eqref{eq: X' not doms Y} holds.
For this same $\secparam$, we also have $\secparam\ge \secparam_1$, hence \eqref{eq: X doms Y} holds as well. By the construction of the distinguisher, when $\D$ is given samples from $\rX$ we have that,

    \begin{align}\label{ineq: D for X}
    &\pr{\D(\rX_{\secparam}^{(1)}, \cdots, \rX_{\secparam}^{(\secparam^d)}; 1^{\secparam}) = 1} - \pr{\D(\rX_{\secparam}^{(1)}, \cdots, \rX_{\secparam}^{(\secparam^d)}; 1^{\secparam}) = 0} \nonumber\\
    ={}&
     \meanDominance{\rX_\secparam}{\rY_\secparam}{\secparam^d}-\meanDominance{\rY_\secparam}{\rX_\secparam}{\secparam^d}\\
    \ge{}& - \frac{1}{\secparam^{c_1}}  \tag{From \Cref{eq: X doms Y}}
    \end{align}

    On the other hand. When the distinguisher has access to samples from $\rhX$, we have

    \begin{align}\label{ineq: D for X'}
    &\pr{\D(\rhX_{\secparam}^{(1)}, \cdots, \rhX_{\secparam}^{(\secparam^d)}; 1^{\secparam}) = 0} - \pr{\D(\rhX_{\secparam}^{(1)}, \cdots, \rhX_{\secparam}^{(\secparam^d)}; 1^{\secparam}) = 1} \nonumber\\
    ={}&
    \meanDominance{\rY_\secparam}{\rhX_\secparam}{\secparam^d} -\meanDominance{\rhX_\secparam}{\rY_\secparam}{\secparam^d}\\
    \ge{}& \frac{1}{\secparam^{c_2}} \tag{From \Cref{eq: X' not doms Y}}
    \end{align}

    Summing up Inequality~(\ref{ineq: D for X}) and Inequality~(\ref{ineq: D for X'}), we have 
    \begin{align*}
        &\left(\pr{\D(\rhX^{(1)}_\secparam, \cdots, \rhX^{(\secparam^d)}_\secparam; 1^{\secparam}) = 0} - \pr{\D(\rX_{\secparam}^{(1)}, \cdots, \rX_{\secparam}^{(\secparam^d)}; 1^{\secparam}) = 0} \right) \\
        +{}&\left(\pr{\D(\rX_{\secparam}^{(1)}, \cdots, \rX_{\secparam}^{(\secparam^d)}; 1^{\secparam}) = 1} - \pr{\D(\rhX^{(1)}_\secparam, \cdots, \rhX^{(\secparam^d)}_\secparam; 1^{\secparam}) = 1} \right) \\
        {}&\geq \frac{1}{\secparam^{c_2}}-\frac{1}{\secparam^{c_1}},
    \end{align*}

    From the last inequality, we see the distinguishing advantage of $\D$ is at least $\frac{1}{2}(\frac{1}{\secparam^{c_2}}-\frac{1}{\secparam^{c_1}})$.  Notice that for $c_1>c_2$ there exists a $\poly{\cdot}$ such that $\frac{1}{2}(\frac{1}{\secparam^{c_2}}-\frac{1}{\secparam^{c_1}})\ge 1/ \poly{\secparam}$. 
    However, we know that such a distinguisher $\D$ cannot exist by definition (\Cref{def: compInd}). Therefore, we conclude that $\rhX \compEmpDom \rY.$  

\end{proof}

The proof of the following proposition follows the exact same reasoning as \Cref{prop: EM-domination and indistinguishability for X}, with only minor modifications. 

\begin{proposition}
    \label{prop: EM-domination and indistinguishability for Y}
    Let $\rX=\{\rX_\secparam\}_{\secparam\in\mathbb{N}}, \rY=\{\rY_\secparam\}_{\secparam\in\mathbb{N}}, \rhY=\{\rhY_\secparam\}_{\secparam\in\mathbb{N}}$ be random variable ensembles, such that $\rY,\rhY$ are computationally indistinguishable. Then $\rX$ computationally dominates $\rY$ if, and only if, $\rX$ computationally  dominates $\rhY$. That is, if $\rY \cong \rhY$ then  $\rX \compEmpDom \rY \Leftrightarrow \rX \compEmpDom \rhY$.
\end{proposition}

\section{Pseudo-Nash in Games that Involve Cryptography}\label{sec: pseudo-for-crypto-games}

In this section, we prove our main theorem, allowing us to translate (pseudo-) Nash equilibria in games that use idealized cryptography into pseudo-Nash in a game where the idealized cryptography is replaced by a (computationally secure) cryptographic protocol.

We start with a brief overview of basic concepts from the literature, defining the security of cryptographic protocols. We have included some intuition and concrete examples in~\Cref{sec:extraprelims}, below, for any interested reader who is unfamiliar with the simulation-based paradigm. The contents of the following subsection would likely be considered common knowledge within the cryptographic protocols community  (and can be safely skipped by a cryptographer), but will help a game theorist understand how to use our result.

\subsection{Add'l Preliminaries on Cryptographic (Composable) Security}\label{sec:extraprelims}
For clarity, we use multi-party cryptographic protocols as our running example, where parties wish to perform some joint computation securely. The common methodology to design and analyze such protocols can be outlined as follows: 

\begin{enumerate}
\item Define the model of computation that specifies the possible actions by the parties, and the ways in which parties can interact with one another. 
\item Specify the goal that the cryptographic protocol aims to achieve. This can be done either by describing properties that the protocol should have, e.g., for encryption, an attacker who observes the ciphertext should gain no information about the plaintext; or by specifying an idealized primitive, which in the cryptographic literature is called the ideal functionality, that embodies the goal that we are aiming to achieve. Consider, for example, the case of cryptographic commitments: a sender (the {\em committer}) with some input $x$ can run an interactive protocol with a receiver (the {\em verifier}) such that the following properties are satisfied: 
\begin{enumerate} 
\item (hiding) At the end of the above protocol, the receiver has not learned anything about $x$. 
\item (binding) The sender is ``committed'' to his input, i.e., there is a way to convince the receiver that $x$ was the sender's original input, but there is no way to lie that his input was some $x'\neq x$. 
\end{enumerate}
The above security goals are embodied by a simple commitment functionality, which upon receiving an input $x$ from the sender records it and informs the receiver that some input was received (without leaking anything about it); and once the input has been recorded and the receiver has been informed about it, the sender can instruct the functionality to reveal it to the receiver. It is easy to verify that the interaction of a protocol with the above functionality satisfies the security properties of a commitment. 
\item Specify the protocol, i.e., define the actions, and prove that it satisfies a cryptographic security definition (as discussed below).  
\end{enumerate} 

\paragraph{\sc Simulation-based security} 
We use the notion of {\em simulation-based security} as the security definition for cryptographic protocols. In addition to being the standard security definition in modern cryptographic protocols analysis,  simulation-based security is also suitable for our goal of seamlessly using cryptography within game theory, due to the composition theorems that it comes equipped with. In a nutshell, such theorems allow us to replace, within a bigger protocol, idealized cryptographic primitives with their cryptographic implementations, without destroying their security properties. This might sound very close to our main theorem. Therefore, it is worth here noticing a highly relevant difference between the standard approaches taken by cryptography and game theory to frame and reason about the (mis)behavior of parties in protocols and games, respectively:  In particular,  cryptography assumes that the adversary might corrupt several parties in a coordinated manner as opposed to rationality which is common in game theory. We note in passing that there are notable refinements of the cryptographic model to allow for capturing non-coordinated definitions~\cite{DBLP:conf/crypto/AlwenKMZ12}; similarly, game-theoretic stability notions that allow coordination have been proposed (see~\cite{10.1007/978-3-642-25280-8_1} for a nice survey.)  

We can use any composable (simulation-based) security cryptographic framework, e.g.,~\cite{DBLP:journals/joc/Canetti00,Canetti_Composable,DBLP:conf/tcc/BackesPW04} but for this exposition, we will focus on Canetti's Universal Composability (UC) framework~\cite{Canetti_Composable} which is the most broadly adopted in the modern literature. 

The model of computation in UC models (the strategies of the) parties as interactive Turing machines (ITMs).\footnote{In UC, the parties are in fact {\em instances} of ITMs, simulated by a single, larger ITM which is the environment. For sake of accessibility to the broader computer science and game theory audience, we refer to them as ITMs here.} These are Turing machines---a standard theoretical abstraction of computing devices---with dedicated input and output tapes (corresponding to I/O interface) and communication tapes that can be connected to other ITMs' communication tapes: If two ITMs $\vitm_1$ and $\vitm_2$  have connected (linked) communication tapes, then $\vitm_1$ can write to its communication tape shared with $\vitm_2$ so that $\vitm_2$ can read it and vice versa.

An interactive protocol among $n$ parties can be described as a vector ${\bf\pi}=(\pi_1,\ldots,\pi_n)$ of $n$ ITM's, where the $i$th ITM, $\pi_i$, describes the (interactive) strategy of the $i$-th party (this can be thought of as a machine that specifies all possible reactions of an agent in an extensive form game). For simplicity, we will assume that the $n$ parties' ITMs are not connected to one another; instead, when in the protocol we want to have Party $i$ send a message to Party $j$, we will assume that they are both connected to another ITM (corresponding to a communication channel between them) that is linked to both of them.

Given the above model, a party that deviates from the protocol can be captured as a party that replaces its protocol-prescribed ITM with a different one. The actual cryptographic security definition considers situations where several parties might deviate in a coordinated (by an adversary) manner. However, for the purpose of this work---since we only consider Nash equilibrium with rational players---it suffices to discuss the definition restricted to any one of the parties being corrupted (and deviating).

The security definition of a cryptographic protocol is given as follows: First, we specify the ideal goal by means of an ideal functionality $\vfunc$ (also an ITM), which as in Step 1, above, embodies the idealized version of our protocols goals, e.g., in the example above, a commitment functionality, which is an ITM with the behavior discussed above, that has communication tapes linked with the committer and the verifier. Then the definition of what it means for a protocol $\vprot$---which, recall, is a profile of ITMs---to {\em securely realize} $\vfunc$ against a single corruption is devised by comparing an execution of the protocol (the so-called {\em real world}) to an ideal invocation of the functionality $\vfunc$ (the so-called ideal world). In a nutshell, we require that there exists no distinguisher (also an ITM) that distinguishes an ideal-world execution from a real-world execution with better than negligible probability. (In the context of UC security, a distinguisher is also called the {\em environment} as it defines the protocol's input/output interfaces.)  

In a nutshell, a cryptographic protocol securely realizes an idealized cryptographic functionality $\vfunc$, if any deviation in the real world can be mapped to an indistinguishable deviation in the ideal world so that the input/output behavior of the two worlds remains indistinguishable.

A bit more formally, consider a distinguisher $\D$ who chooses inputs to the protocols ITMs and gets to observe their outputs. For any such $\D$, denote by $V^{\D,Real}_{{\bf\pi}=(\pi_1,\ldots,\pi_n)}$ the random variable corresponding to the view of $\D$ when interacting with the protocol. (This view includes $\D$'s own input and randomness, the security parameter, and all inputs and outputs to the ITMs in ${\bf\pi}$.) Similarly, $V^{\D,\vfunc,Ideal}_{(\phi_1,\ldots,\phi_n)}$ corresponds to the view of $\D$ in an ideal invocation of $\vfunc$, where each $\phi_i$ is the ``dummy'' ITM that simply forwards inputs from $\D$ to $\vfunc$ and hands any received inputs back to $\D$. Note that all ITMs involved in the above definition (including $\D$) are assumed to be probabilistic polynomial time (PPT).

We then say that protocol $\bf\pi$, is secure if there exists an ITM $\Sim$, called the {\em simulator}, such that for any $i\in[n]$, and any PPT ITM $\pi_i'$ the view $V^{D,Real}_{(\pi_i',\pi_{-i})}$ of $D$ when interacting with ${\bf\pi}$ where party $i$ plays $\pi_i'$ instead of $\pi_i$ is (computationally) indistinguishable  from the view $V^{D,Ideal}_{(\Sim(\pi_i'),\phi_{-i})}$ of $D$ in an ideal invocation where party $i$ plays the simulator's strategy (instead of the ``dummy'' one).

\subsection{Normal Form Computational Games}\label{sec:nfcg}

Combining the above definition---which allows to replace deviations in the cryptographic protocols with indistinguishable deviations in the idealized-crypto setting---with \Cref{prop: EM-domination and indistinguishability for X} 
already hints towards the proof of our main theorem. However, we still need to formally define games among ITMs which will allow us to use cryptography. We do so in the following: 

\begin{definition} [Normal Form Computational Game]
    A real normal form computational game $\game$ consists of the following:
    \begin{itemize}
        \item $[n] = \{1, \cdots, n\}$ is the set of players in the game. $\Omega$ is the set of all possible outcomes of the game.
        $\secparam$ is the security parameter.
        \item Players' strategies are interactive, probabilistic polynomial time (PPT) Turing machines. We use  $\cN_i = \{1, 2, \cdots\}$ to index these strategies/machines. That is, player $i$ chooses strategy $s \in \cN_i$ means that she is playing using the $s^\text{th}$ Turing machine. This Turing machine receives a single input: the security parameter, written in unary: $1^\secparam$. We write $s_i^{\secparam}$ for the strategy/Turing machine chosen by player $i$, executed with input $1^\secparam$. Let $\cN = \cN_1\times\ldots\times\cN_n $.

        \item The Turing machines chosen by the players can be linked via one or more Turing machines, which are fixed and part of the game description. (These will correspond to  ideal functionalities, or simple communication channels.) The security parameter $\secparam$ is written on the input tape of all ITMs in the game. 
        All players simultaneously pick their strategies (Turing machine $s_i \in \cN_i$ for every player $i$), and these Turing machines interact with each other according to the rules of the game, and a single outcome $\omega \in \Omega$ is selected (and learned by the players).

        \item $U_i: \Omega  \rightarrow \mathbb{R}$ is the random variable that indicates the utility of player $i$. We slightly overload notation and write $U_i^{\secparam}(s)$ for the utility of player $i$ when the strategy profile is $s=(s_1,\ldots,s_n)$ and the security parameter is $\secparam$. We write $U_i^{\secparam} (\hs_i; s_{-i})$  when player $i$ unilaterally deviates to $\hs_i$ and the rest play according to $s_{-i}$, and assume  $U_i^{\secparam}$ is efficiently computable in $\secparam$.
    \end{itemize}

\end{definition}

\subsection{Cryptographic Pseudo-Nash}
\label{sec: ideal pseudo implies cryptographic pseudo}

In this subsection, we formalize how pseudo-equilibria in a game that uses an ideal functionality $\F$ automatically translate to pseudo-equilibria
in the corresponding game where $\F$ is replaced by a cryptographic implementation $\Pi$.

\begin{definition}[Game implementation]\label{def: game implementation}
Let $\idealGame([n],\cN^{\F},U)$ be a normal-form computational game in which strategies may invoke an ideal functionality $\F$.
Let $\Pi$ be a protocol that securely realizes $\F$.
The \emph{implementation} of $\idealGame$ with $\Pi$, denoted $\realGame([n],\cN^{\Pi},U)$, is defined by a mapping
$\Comp_i:\cN_i^{\F}\to \cN_i^{\Pi}$ for each player $i$, where $\Comp_i(\sigma_i)$ executes $\sigma_i$ while replacing each call to $\F$ with an invocation of $\Pi$.
The outcome space $\Omega$ and utilities $U_i:\Omega\to\mathbb{R}$ are the same in $\idealGame$ and $\realGame$, and utilities are evaluated on the induced outcome distribution.
\end{definition}

\begin{theorem}\label{thr: pseudo implies pseudo}
Let $\F$ be an ideal functionality and $\Pi$ be a protocol that securely realizes $\vfunc$ against a single corruption. Let $\idealGame([n],\cN^\F,U)$ be a computational game with access to $\F$, and  $\realGame([n],\cN^\Pi,U)$ be its implementation with $\Pi$ (\Cref{def: game implementation}). If $\sigma\in \cN^\F$ is a computational pseudo-Nash equilibrium (\Cref{def: computationalPseudo}) of $\idealGame$, then the corresponding strategy profile $s=(\Comp(\sigma_1),\ldots, \Comp(\sigma_n))\in\cN^{\Pi}$is a computational pseudo-Nash equilibrium of $\realGame$.
\end{theorem}

\begin{proof}
Let $U_i(\sigma_i,\sigma_{-i})$ be the utility random variable ensemble of player $i$ when players playing strategy $\sigma_i$ in $\idealGame$. Similarly, let $U_i(s_i,s_{-i})=\ensembleU{U_i}{(s^*_i,s^*_{-i})}$ be the utility random variable ensemble of player $i$ when players playing strategy $s_i$ in $\realGame$. For ease of notation, we always use $\sigma$ to denote strategy profiles in $\idealGame$ and $s$ for strategy profiles in $\realGame$. 

We prove the theorem by contradiction. Suppose that $s=(s_1,\ldots,s_n)$ is not a computational pseudo-Nash equilibrium in $\realGame$. By \Cref{def: computationalPseudo}, this implies there exists a player $i$ and unilateral deviation $\hs_i\in\cN^\Pi$ in $\realGame$ such that the utility ensemble of $s$ does not computationally mean dominate the utility ensemble of the deviation, $\neg(U_i(s_i,s_{-i}) \compEmpDom U_i(\hs_i,s_{-i}))$, let  $\hs=(\hat{s}_i, s_{-i})$. This means that strategy $\hs_i$ has observably better utility for player $i$. Towards the contradiction, we will show that if such a strategy exists in the real game $\realGame$ then there exists a strategy in the ideal game $\idealGame$ that also has observably better utility.  By definition, this implies that $\sigma$ is not a pseudo-equilibrium profile.

By the simulation-based security of $\Pi$, for every PPT unilateral deviation $\hat{s}_i$ in the real game $G_{\Pi}$ there exists a PPT ideal-world strategy (simulator) $\Sim(\hat{s}_i)$ such that the real execution under $(\hat{s}_i,s_{-i})$ is simulated in the ideal game under $(\Sim(\hat{s}_i),\sigma_{-i})$. More formally, fix any PPT distinguisher $\D$. For any strategy profile $t$ in the real game $\realGame$ (resp., $\tau$ in the ideal game $\idealGame$), let $V^{\D,Real}_t$ (resp., $V^{\D,Ideal}_\tau$) denote the random variable corresponding to $\D$'s view in the execution induced by that profile. In particular, $V^{\D,Real}_{(\hat{s}_i,s_{-i})}$ is $\D$'s view in the real execution where player $i$ deviates to $\hat{s}_i$, and $V^{\D,Ideal}_{(\Sim(\hat{s}_i),\sigma_{-i})}$ is $\D$'s view in the corresponding ideal execution where player $i$ follows $\Sim(\hat{s}_i)$ Similarly, $V^{\D,Real}_{(s_i,s_{-i})}$ and $V^{\D,Ideal}_{(\sigma_i,\sigma_{-i})}$ denote $\D$'s views under the honest profiles in the real and ideal executions. Thus,
\begin{align*}
V^{\D,Real}_{(\hat{s}_i,s_{-i})} \cong V^{\D, Ideal}_{(\Sim(\hat{s}_i),\sigma_{-i})} \qquad \text{and} \qquad V^{\D,Real}_{(s_i,s_{-i})} \cong V^{\D, Ideal}_{(\sigma_i,\sigma_{-i})}.
\end{align*}

From the indistinguishability of the views, we get that the corresponding utility ensembles are also indistinguishable, since $U_i$ can be efficiently computed from $\D$'s view. That is 
\[ U_i(\hat{s}_i,s_{-i}) \cong U_i(\Sim(\hat{s}_i),\sigma_{-i})\qquad \text{and} \qquad U_i(s_i,s_{-i}) \cong U_i(\sigma_i,\sigma_{-i}).\]

Next, we apply the invariance of computational mean dominance under indistinguishable replacements. Given that $\neg (U_i(s_i, s_{-i}) \gtrsim U_i(\hat{s}_i, s_{-i}))$, we substitute the real deviation with its ideal counterpart. From \Cref{prop: EM-domination and indistinguishability for Y} we get $\neg \left( U_i(s_i, s_{-i}) \gtrsim U_i(\Sim(\hat{s}_i), \sigma_{-i}) \right).$ Similarly, we substitute the real honest execution with the ideal honest execution. Then, by applying  \Cref{prop: EM-domination and indistinguishability for X}, we get $\neg \left( U_i(\sigma_i, \sigma_{-i}) \gtrsim U_i(\Sim(\hat{s}_i), \sigma_{-i}) \right).$

The last statement implies that $\hat{\sigma}_i = \Sim(\hat{s}_i)$ is a unilateral deviation in the ideal game $\idealGame$ that violates the pseudo-Nash condition for $\sigma,$ contradicting that $\sigma$ is a computational pseudo-Nash equilibrium in $G_{\F}$ Therefore, our assumption was false, and $s$ is a computational pseudo-Nash equilibrium in $\realGame.$ 

\end{proof}

\section*{Acknowledgments} The authors would like to thank Joel Alwen and Georgios Amanatidis for their interesting discussions and their input in the early stages of this project. Vassilis Zikas was supported in part by NSF Awards No. 2448339 and No. 2531010, JPMorgan Chase, the AI Security Institute, the Stellar Development Foundation, and Sunday Group, Inc. Yu Wei was supported in part by NSF Awards No. 2448339 and No. 2531010, the AI Security Institute, and Sunday Group, Inc. Alexandros Psomas and Athina Terzoglou were supported in part by NSF CAREER Award CCF-2144208 and research awards from Halcyon Futures, Google, and Supra. Work was done in part while Vassilis Zikas and Alexandros Psomas were visiting the Simons Institute for the Theory of Computing, UC Berkeley. 
\clearpage

\bibliographystyle{splncs04}
\bibliography{refs}
\clearpage

\appendix
\section{Missing proofs}\label{app:equiv-nash-proofs}

\begin{proof}[Proof of~\Cref{lemma: E[X]>E[Y]impliesComDominance}]
    \label{proof: lemma: E[X]>E[Y]impliesComDominance}
    Let us define $\Delta = \E{}{\rX}-\E{}{\rY}>0$. We want to show that $\rX\compEmpDom\rY$, according to~\Cref{def:computational-mean-dominance}. That is, for all constants $c \geq 1$, there exists a $\ch>c$, which we set to $\ch=4c$, and a $\secparam_0>0$ such that for all $\secparam\ge\secparam_0$: 
    \begin{equation}\label{eq: compDomEnsemble}
        \meanDominanceEQ{\rY_\secparam}{\rX_\secparam}{\secparam^{4c}}-\meanDominanceEQ{\rX_\secparam}{\rY_\secparam}{\secparam^{4c}} < \frac{1}{\secparam^{c}}\\
    \end{equation}
    
    where $\rX_{\secparam}^{(j)},\rY_{\secparam}^{(j)}$ are i.i.d. samples from $\rX_\secparam,\rY_\secparam$ respectively. 
    
    Since $\rX$ and $\rY$ are random variables (that is $X_\secparam=X$ and $Y_\secparam=Y$ for all $\secparam$),  \Cref{eq: compDomEnsemble} is reduced to:
    \begin{equation}\label{eq: compDomLemma}
        \meanDominanceEQ{\rY}{\rX}{\secparam^{4c}}-\meanDominanceEQ{\rX}{\rY}{\secparam^{4c}} < \frac{1}{\secparam^{c}}\\
    \end{equation}

For a fixed $\secparam$ and constant $c$, define the empirical averages and apply Hoeffding's inequality~\cite{Books:H63}, for any $\epsilon>0$ we have

\[
\pr{\abs{\frac{1}{\secparam^{4c}} \sum_{j=1}^{\secparam^{4c}}\rX^{(j)} - \Ex{\rX}}\ge\epsilon}  \le 2\exp\left( -\frac{2\secparam^{4c}\epsilon^2}{R^2} \right)
\]
and
\[
\pr{\abs{\frac{1}{\secparam^{4c}} \sum_{j=1}^{\secparam^{4c}}\rY^{(j)} - \Ex{\rY}}\ge\epsilon}  \le 2\exp\left( -\frac{2\secparam^{4c}\epsilon^2}{R^2} \right)
\]

For ease of notation, let us define for fixed but arbitrary $\secparam,c$:
\[
\overline{\rX} = \frac{1}{\secparam^{4c}} \sum_{j=1}^{\secparam^{4c}}\rX^{(j)} \quad\text{and}\quad \overline{\rY} = \frac{1}{\secparam^{4c}} \sum_{j=1}^{\secparam^{4c}}\rY^{(j)}
\]

Set $\epsilon = \Delta/4$. We consider the following two events, $\overline{\rX}  \ge \Ex{\rX} - \Delta/4$ and $\overline{\rY}  \le \Ex{\rY} + \Delta/4$. When both events occur we have that $\overline{\rX} -  \overline{\rY} \ge \Ex{\rX} - \Delta/4 - ( \Ex{\rY} + \Delta/4) \ge \Delta/2 > 0$

\begin{align*}
    \pr{\overline{\rX}>  \overline{\rY}} &\ge \pr{\left(\overline{\rX}  \ge \Ex{\rX} - \Delta/4 \right) \;\&\; \left(\overline{\rY}  \le \Ex{\rY} + \Delta/4\right)}\\
    &= \pr{\left(\overline{\rX}  \ge \Ex{\rX} - \Delta/4 \right)}\cdot \pr{ \left(\overline{\rY}  \le \Ex{\rY} + \Delta/4\right)}\\
    &\ge\left(1-2\exp\left( -\frac{2\secparam^{4c}(\Delta/4)^2}{R^2} \right)\right) \left(1-2\exp\left( -\frac{2\secparam^{4c}(\Delta/4)^2}{R^2} \right)\right)\\
    &\ge 1- 4\exp\left( -\frac{\secparam^{4c}\Delta^2}{8R^2} \right)
\end{align*}

To conclude the proof, we need to show that there exists a $\secparam_0$ such that for all $\secparam\ge\secparam_0$ and for all $c\ge1$ \Cref{eq: compDomLemma} holds.  We can see that it suffices to find the number of samples required to obtain $\pr{\overline{\rX} > \overline{\rY}}\ge\frac{1}{2}$.  That is because $\pr{\overline{\rY}>\overline{\rX}}-\pr{\overline{\rX}>\overline{\rY}} \le 0 \le \frac{1}{\secparam^c}.$ Thus, we need

\begin{equation*}
    4\exp\left(-\frac{\secparam^{4c}\Delta^2}{8R^2}\right) \le \frac{1}{2} 
    \Leftrightarrow -\frac{\secparam^{4c}\Delta^2}{8R^2} + \ln 4 \le  -\ln2
    \Leftrightarrow \secparam^{4c} \ge \frac{8R^2}{\Delta^2}\ln8
\end{equation*}   
Thus, since $c\ge1$ it suffices to set $\secparam_0 = \left(\frac{8R^2}{\Delta^2}\ln8\right)^{1/4}.$

\end{proof}

\begin{proof}[Proof of~\Cref{lemma: E[X]=E[Y]impliesComDominance}]
 \label{proof: lemma: E[X]=E[Y]impliesComDominance}
 For two random variables $\rX,\rY$ with $\E{}{\rX}=\E{}{\rY}$, we want to show that $\rX\compEmpDom\rY$ and $\rY\compEmpDom\rX$. Since $\rX\compEmpDom\rY$ and $\rY\compEmpDom\rX$ are symmetric, without loss of generality we will show that $\rX\compEmpDom\rY$. That is, for all constants $c \geq 1$, there exist $\ch>c$ and a $\secparam_0>0$ such that for all $\secparam\ge\secparam_0$: 
    \begin{equation}\label{eq: compDomEnsemble X=Y}
        \meanDominanceEQ{\rY}{\rX}{\secparam^{\ch}}-\meanDominanceEQ{\rX}{\rY}{\secparam^{\ch}} < \frac{1}{\secparam^{c}}\\
    \end{equation}
    
    where $\rX^{(j)},\rY^{(j)}$ are i.i.d. samples from $\rX,\rY$ respectively.
 
 We will consider the random variable $\rZ = \rX - \rY$. By linearity of expectations we have that $\Ex{\rZ} = \Ex{\rX} -\Ex{\rY} = 0$.  For some fixed but arbitrary $\secparam,\hc$ let the empirical mean of $\rZ$ with $\secparam^{\ch}$ i.i.d. samples be
    \[
    \overline{\rZ}_{\secparam,\hc} =  \frac{1}{\secparam^{\hc}}\sum_{j=1}^{\secparam^{\hc}} \rZ^{(j)}=\frac{1}{\secparam^{\hc}} \sum_{j=1}^{\secparam^{\hc}} \left(\rX^{(j)}- \rY^{(j)}  \right) .
    \]
    
    Since the random variable $\rZ$ has finite support, we know that its moments are also finite. Let $\sigma^2 = \Ex{\rZ^2}<\infty $ and  $\rho = \Ex{\abs{\rZ}^3}<\infty$. If $\sigma^2=0$, then $X$ and $Y$ are constants. Since $\Ex{X}=\Ex{Y} $, they are the same constant, and their empirical means are always equal, satisfying the computational condition trivially. If $\sigma^2>0$, by the Berry–Esseen theorem~\cite{berry1941accuracy,esseen1942liapunov}, considering the cumulative distribution function (CDF) $F_{\secparam,c}(\cdot)$ of $\frac{\overline{\rZ}_{\secparam,\ch} \sqrt{\secparam^{\ch}}}{\sigma}$ and the CDF of normal distribution $\Phi(\cdot)$ we have that for a $C<0.47 $:
    \begin{equation}\label{eq: Berry-Esseen X=Y}
        \sup_{x\in\bR} \abs{F_{\secparam,c}(x)-\Phi(x)} \le \frac{C\rho}{\sigma^3\sqrt{\secparam^{\hc}}} 
    \end{equation}

    To satisfy \Cref{eq: compDomEnsemble X=Y} we show that for all constants $c$ and a $\ch>c$ there exists $\secparam_0$ such that for all $\secparam\ge\secparam_0$ we have $\pr{\overline{\rZ}_{\secparam,\hc}<0} - \pr{ \overline{\rZ}_{\secparam,\ch}>0} <\frac{1}{\secparam^c}$. Notice that we can remove the event that $\overline{\rZ}_{\secparam,\ch}=0$ from both probabilities without changing their difference. From \Cref{eq: Berry-Esseen X=Y} we have that $\pr{\overline{\rZ}_{\secparam,c}<0}\le\pr{\overline{\rZ}_{\secparam,c}\le0} \le \frac{C\rho}{\sigma^3\sqrt{\secparam^{\hc}}}   + \Phi(0) = \frac{C\rho}{\sigma^3\sqrt{\secparam^{\hc}}}  +\frac{1}{2}$. We can easily show that

    \begin{equation*}
        \pr{ \overline{\rZ}_{\secparam,\ch}>0} = 1 - \pr{\overline{\rZ}_{\secparam,\hc}\le0} \ge 1- \left(\frac{1}{2} + \frac{C\rho}{\sigma^3\secparam^{2c}} \right) = \frac{1}{2} - \frac{C\rho}{\sigma^3\secparam^{2c}} 
    \end{equation*}

    Plugging these inequalities in to \Cref{eq: compDomEnsemble X=Y} we get:
    \begin{equation*}
        \pr{\overline{\rZ}_{\secparam,\hc}<0} - \pr{ \overline{\rZ}_{\secparam,\ch}>0}  \le \frac{1}{2} + \frac{C\rho}{\sigma^3\sqrt{\secparam^{\hc}}}  - \left(\frac{1}{2} - \frac{C\rho}{\sigma^3\sqrt{\secparam^{\hc}}}  \right) \le\frac{2C\rho}{\sigma^3\sqrt{\secparam^{\hc}}}  
    \end{equation*}

    Therefore, we want that for all $c\ge1$

    \begin{equation*}
        \frac{2C\rho}{\sigma^3\sqrt{\secparam^{\hc}}}  < \frac{1}{\secparam^c} \implies \secparam^{\frac{\ch}{2} -c} > \frac{2C\rho}{\sigma^3} 
    \end{equation*}
    
    By selecting $\ch=2c+2$, we get that   
    \begin{equation*}
         \secparam > \frac{2C\rho}{\sigma^3}
    \end{equation*}

That concludes the proof. 

\end{proof}

\begin{proof}[Proof of~\Cref{lemma: XdomY implies E[X]>=E[Y]}]
\label{proof: lemma: XdomY implies E[X]>=E[Y]}
We prove by contradiction. By the fact that $\rX\compEmpDom\rY$ we have that for all constants $c \geq 1$, there exist $\ch>c$ and a $\secparam_0>0$ such that for all $\secparam\ge\secparam_0$: 
    \begin{equation}\label{eq: compDomEnsemble exp}
        \meanDominanceEQ{\rY}{\rX}{\secparam^{\ch}}-\meanDominanceEQ{\rX}{\rY}{\secparam^{\ch}} < \frac{1}{\secparam^{c}}\\
    \end{equation}
    
    where $\rX^{(j)},\rY^{(j)}$ are i.i.d. samples from $\rX,\rY$ respectively.

For the rest of the proof, it suffices to set $\ch=4c$.
For ease of notation, let us define for fixed but arbitrary $\secparam,c$ the empirical mean random variable of $\rX,\rY$:
\[
\overline{\rX} = \frac{1}{\secparam^{4c}} \sum_{j=1}^{\secparam^{4c}}\rX^{(j)} \quad\text{and}\quad \overline{\rY} = \frac{1}{\secparam^{4c}} \sum_{j=1}^{\secparam^{4c}}\rY^{(j)}
\]
Assume that $\Ex{Y}>\Ex{X}$.  Let $\Delta=\Ex{Y}-\Ex{X}>0$ We will show that for all $c\ge1$ there exists a $\secparam_1$ such that for all $\secparam\ge\secparam_1$ \Cref{eq: compDomEnsemble exp} does not hold. Notice that~\Cref{lemma: E[X]>E[Y]impliesComDominance} does not directly contradict~\Cref{eq: compDomEnsemble exp} since the lemma implies 
$\pr{\overline{\rY}>\overline{\rX}} -\pr{\overline{\rX}>\overline{\rY}} \ge \frac{1}{\secparam^c}.$ We need to show a tighter bound.  

 Towards the contradiction, we will show that $\pr{\overline{\rY}>\overline{\rX}} -\pr{\overline{\rX}>\overline{\rY}}  \ge \frac{2}{\secparam^c}\ge \frac{1}{\secparam^c}.$ Following the same steps as in~\Cref{lemma: E[X]>E[Y]impliesComDominance}.

By Hoeffding's inequality, for any $\epsilon>0$  we have:
\[
\pr{\abs{\overline{\rX} - \E{}{\rX}}\ge\epsilon} \le 2 \exp\left(-\frac{2\secparam^{4c}\epsilon^2}{R^2}\right)\]
and
 \[\pr{\abs{\overline{\rY} - \E{}{\overline{\rY}}}\ge\epsilon} \le 2 \exp\left(-\frac{2\secparam^{4c}\epsilon^2}{R^2}\right)\]
Setting $\epsilon = \Delta/4$, we have that,

\begin{align*}
   \pr{\overline{\rY}>\overline{\rX}}  &\ge  \pr{\overline{\rY}\ge \E{}{\rY} - \Delta/4 }\cdot\pr{\overline{\rX}\le \E{}{\rX} + \Delta/4 }\\
    &\ge \left(1-  2 \exp\left(-\frac{\secparam^{4c}\Delta^2}{8R^2}\right)\right)^2 \ge 1 - 4 \exp\left(-\frac{\secparam^{4c}\Delta^2}{8R^2}\right)
\end{align*}
To obtain a tighter to contradict~\Cref{eq: compDomEnsemble exp} bound we need to find a $\secparam_0$ such that for all $\secparam\ge\secparam_0$ and for all $c$ the $\pr{\overline{\rY}>\overline{\rX}} \ge1/2+1/\secparam^c.$  That implies 
$\pr{\overline{\rX}>\overline{\rY}} \le1/2-1/{\secparam^c}$. Combining the two we get that  $\pr{\overline{\rX}>\overline{\rY}}-\pr{\overline{\rY}>\overline{\rX}}  \le -\frac{2}{\secparam^c}$ which directly contradicts~\Cref{eq: compDomEnsemble exp}. We can find the number of samples required by

\begin{equation*}
   \pr{\overline{\rY}>\overline{\rX}}\ge  1 - 4 \exp\left(-\frac{\secparam^{4c}\Delta^2}{8R^2}\right) \ge1/2+\frac{1}{\secparam^c} \Leftrightarrow 4\exp\left(-\frac{\secparam^{4c}\Delta^2}{8R^2}\right) \le \frac{1}{2} - \frac{1}{\secparam^c} 
\end{equation*}   

We can easily see that $4\exp\left(-\frac{\secparam^{4c}\Delta^2}{8R^2}\right)$ is a decreasing function of $\secparam$ that goes to 0 as $\secparam$ goes to infinity. On the other hand $ \frac{1}{2} - \frac{1}{\secparam^c} $ is an increasing function of $\secparam$  that goes to $1/2$ as $\secparam$ goes to infinity. Furthermore, because of the monotonicity, if the inequality holds for some constant  $c$ it also holds for all $c'>c$. Thus, it suffices to show it for $c=1$. Finally, to show there exists a $\secparam_0$, consider any constant $\epsilon\in(1/2)$ such that $ 4\exp\left(-\frac{\secparam^{4}\Delta^2}{8R^2}\right) < \epsilon $ and $\epsilon < 1/2-1/\secparam$. Since we want both to be satisfied simultaneously we get
\[\secparam_0 > \max\left\{ \frac{8R^2}{\Delta^2}\ln\left(\frac{4}{\epsilon}\right) , \frac{2}{1-2\epsilon}\right\}.\]
That concludes the proof of \Cref{lemma: XdomY implies E[X]>=E[Y]}.

\end{proof}
\clearpage
\section{Examples}\label{apx: examples}

\begin{example}[Games on a Ledger]\label{example: nothing works for behavioral}
    Consider the problem of a blockchain miner. A blockchain implements the distributed ledger functionality, which, loosely speaking, provides the following functionality\footnote{We note that we oversimplify the functionality, in order to make our point. Examples of ledger functionalities implemented by mainstream blockchain protocols like Bitcoin and Ouroboros/Cardano---and corresponding proofs of security---can be found in~\cite{C:BMTZ17,CCS:BGKRZ18}.}: every party can provide inputs to the ledger, which is recorded on a future block; every party can receive the ledger's current state. The blockchain implements the ledger functionality assuming an honest majority of the participants, or of a resource held by the participants like hashing power (in proof of work) or stake (in proof of stake), i.e., the miners. Given an ideal ledger, consider a simple game between miners, where a miner gains a large reward $R$ when doing a task quickly, and a small reward $r$ when doing it slowly, where ``time'' is measured by block timestamps. For the sake of exposition, assume that the miner pays a large cost when doing the task quickly, but no cost when doing the task slowly, and thus the miner's best response under these rules is to do the task slowly. In the context of DeFi (Decentralized finance),~\cite{yaish2022blockchain} proves that, in certain blockchain protocols, it is possible for miners to deviate in a way that they can manipulate the timestamps of blocks. For our simple game, this would mean that the miner could claim the large reward $R$, contradicting the prediction in the ideal game.
    
\end{example}

\begin{example}[The Guessing Game]\label{example: Nash fails}
    Consider a simple, ideal two-player game, where player $1$ commits to $\kappa$ bits, and player $2$ needs to guess them. If player $2$ guesses all $\kappa$ bits, then she wins and gets a utility of $1$; otherwise, she loses and gets a utility of $-1$.
    An ideal commitment scheme $\textsc{C}$, for the case of two players, works as follows. Player $1$ sends a value $x$ to $\textsc{C}$, which returns a ``receipt'' $c$ to player $2$. Later, player $1$ sends ``open'' to $\textsc{C}$, and $\textsc{C}$ sends $x$ to player $2$. In the ideal world, $\textsc{C}$ is \emph{perfectly binding} (player $1$ cannot open the commitment $\textsc{C}$ to any $x' \neq x$) and \emph{perfectly hiding} (player $2$ cannot learn about $x$ from $c$). It is easy to see that, when one uses an ideal commitment scheme in the game described, in the unique Nash equilibrium of the game, both players pick an $\kappa$-bit string uniformly at random (and player $2$ wins with probability $1/2^\kappa$).

    In the ``real game,'' the commitment scheme can be implemented in various ways. Perfectly hiding and perfectly binding schemes are known to not exist, but, for example, the ElGamal-based commitment scheme is \emph{computationally hiding} and perfectly binding. Concretely, a ``real-world'' cryptographic commitment scheme consists of a function $Commit(x,r)$ that takes as input a message $x$ and a random value $r$, and outputs a commitment $c$; then ``open'' can be implemented by simply revealing $(x,r)$. Computationally hiding then means that, for $r$ and $r'$ chosen uniformly at random, the distributions of $Commit(x,r)$ and $Commit(x',r')$ are computationally indistinguishable, i.e., no efficient algorithm can distinguish which one is which with probability negligibly better than a random coin-toss.
    Importantly, the above indistinguishability statement does not exclude that an inefficient adversary, who, e.g., performs a brute-force attack by checking all values of $x$ and $r$ against $c$, can guess $x$ with good probability. In fact, when the commitment scheme is perfectly binding, such brute-forcing is always possible. 
    Thus, player $2$'s best response to every strategy of player $1$ is to attempt to guess $(x,r)$ by brute-forcing. And, even imposing computational constraints on player $2$, brute-forcing will still be 
    a strategy that succeeds with a probability $\delta$ better than random guessing. 
    Therefore, picking a uniformly random $\kappa$-bit string is not a best-response for player $2$. And, by making the rewards sufficiently large (e.g.,  infinite), it is not even an $\epsilon$ best-response, for any $\epsilon >0$.
\end{example}

\begin{example}[Auctions with Information Advantage]\label{example: auction where ideal vs real is weird}
    Consider running a sealed-bid second-price auction for a single item:  each bidder submits a {\em bid}  privately to the auctioneer; the winner is the highest bidder, and the payment is the second-highest bid. Sealed-bid second-price auctions have been extensively studied in game theory and mechanism design; one of their greatest appeals is {\em dominant strategy incentive compatibility} (truthfulness): submitting a bid equal to the true valuation is a dominant strategy for every bidder (i.e., weakly better than every other strategy, no matter what all other bidders are doing).
     
     Bidder $i$'s utility for the item is their private value $v_i$ minus the price they are asked to pay $p_i$ if they win and 0 otherwise. Additionally, bidder $i$ gets utility by knowing the identity of the highest valued bidder ($i^*=\arg\max_i v_i$); this scenario can occur, for example, if having information about who won the auction offers agents a competitive advantage in future auctions. Concretely, after the auction finishes, bidder $i$ gets a large reward $R$ with probability $q_i = \max\{ \beta_{i,i^*} - \frac{1}{n-1}, 0 \}$, where $\beta_{i,j}$ is bidder $i$'s posterior probability that $j$ is $i^*$ ($q_i$ can be interpreted as how much better $i$ can do, compared to a random guess).

    In our example, sealed bids are implemented by (ideal) commitment schemes. The auctioneer additionally applies a random permutation to the set of commitments before making them publicly available, to keep bidder identities anonymous. After collecting the openings, the auctioneer privately notifies each bidder whether they won (and if so, the price) and publicly announces the second-highest bid (price) by opening the corresponding (permuted) commitment. Revealing only the price is the minimal disclosure needed for bidders to verify the payment rule, while keeping the winner’s identity and the remaining bids hidden. In the ideal world, this anonymity is realized by an ideal random permutation; in the real world, it is implemented by a pseudorandom permutation ($\PRP$). 
    
    Under an ideal random permutation, the opening step provides no additional identity information beyond the random guess, since it is related to the second-highest bidder. Consequently, apart from the winner knowing her own identity, no one else can identify the highest bidder with probability better than $1/(n-1)$, so the additional reward is strategically irrelevant. Thus, bidders simply bid their true values, and assign $\beta_{i,j}=1/(n-1)$ for every $j\neq i$.

In the real world, under a cryptographic implementation, a first observation is that for bidders playing truthfully, attempting to  ``break" (i.e., invert) the $\PRP$ does not improve the bidder's utility since the only publicly opened commitment is the one corresponding to the second-highest bid (the price), which is unrelated to the identity of the highest-value bidder. However, for a high enough reward $R$, a profitable non-truthful strategy exists if the bidder can also invert the PRP: submit an extremely high bid to win the auction, and therefore learn the true highest value among other bids (since this is now the payment) together with the (permuted) position of the opened bid, which was previously unknown. Thus, even a negligible-probability inversion of the $\PRP$ can deanonymize the commitment and yield a strict increase in expected utility. Therefore, for large enough $R$, the Nash equilibrium of the real game has players bidding non-truthfully (while attempting to break crypto,) despite truth-telling being the equilibrium prediction in the ideal world.

\end{example}
\clearpage
\section{Applications}
In this section, we showcase our definition and our main theorem in three ways. First, we revisit the Guessing Game (\Cref{example: Nash fails}) to illustrate the fundamental mismatch between classical Nash (and computational Nash) and PPT strategies. We use this example to demonstrate how pseudo-Nash correctly ignores ``crypto-breaking'' deviations that, while occurring with negligible probability, can explode expected utility and disqualify an otherwise stable strategy from being a (computational) Nash equilibrium. Second, we formalize the auction example sketched in the introduction (\Cref{example: auction where ideal vs real is weird}), showing that the Nash equilibrium of the real can change drastically, even though in the ideal game the Nash equilibrium is exactly what someone expects. However, our theorem implies that the ideal equilibrium is preserved as a pseudo-Nash equilibrium in the real game. Finally, we revisit one of the first results in game-theoretic cryptography concerning rational secret sharing~\cite{halpern2004rational}, to show how pseudo-equilibrium avoids counter-intuitive impossibility phenomena that hinge on negligible-probability events. Moreover, unlike $\epsilon$-Nash and computational Nash equilibrium~\cite{halpern2016computational}, pseudo-equilibrium remains insensitive to the magnitude of the utility associated with such events.

\subsection{Guessing Game}
\label{sec: guessing game}

We are revisiting a slightly modified version of the two-player game presented in \Cref{example: Nash fails}. The game is zero-sum, player 1 commits to $\secparam$ bits, and player 2 needs to guess the string. For every bit correctly guessed, players 2's utility doubles. Formally, we consider the games shown in \Cref{table: Guessing Game}.

\begin{table}[ht]
\begin{tabularx}{\textwidth}{|>{\setlength\hsize{1\hsize}\setlength\linewidth{\hsize}}X|>{\setlength\hsize{1\hsize}\setlength\linewidth{\hsize}}X|}
\hline \rule{0pt}{4ex}
\textbf{Ideal Game} & \textbf{Real Game}\\
\hline
 \begin{itemize} 
    \item \textbf{Commit Phase:}  
    Player~1 selects $x\in\{0,1\}^{\secparam}$ and commits it using an ideal commitment.
    
    \item \textbf{Guess Phase:}  
    Player~2 chooses a guess $y\in\{0,1\}^{\secparam}$.
    
    \item \textbf{Reveal Phase:}  
    Player 1 can open the commitment to reveal $x$. If player 1 does not open the commitment, she receives utility $-2^\kappa$ ($u_2=2^\kappa).$
\end{itemize} &
\begin{itemize}
    \item \textbf{Commit Phase:}  
    Player~1 chooses $x\in\{0,1\}^{\secparam}$ and a random string $r$, then sends to Player~2 $c = \textsf{Commit}(1^\secparam, x; r)$.
    \item \textbf{Guess Phase:}  
    Player~2 outputs a guess $y$ based on the received $c$.
    \item \textbf{Reveal Phase:}  
    Player~1 sends $(x,r)$ to open the commitment. If $\textsf{Open}(1^\secparam,x,r)$ verifies, the utilities are assigned.
\end{itemize}
 \\\hline
\multicolumn{2}{|c|}{%
  \rule[-4ex]{0pt}{9ex}
  \parbox{0.95\textwidth}{\centering
    For both games, the utilities of the players are $u_2=-u_1=2^\ell$\\
    (where $\ell$ is the number of correctly matched bits).
  }%
}\\\hline
\end{tabularx}
\caption{Guessing Game}\label{table: Guessing Game}
\end{table}
In the ideal game, it is easy to show that the strategy profile where $x,y$ are chosen uniformly at random is a Nash equilibrium. To see this, we will show that there is no profitable deviation for player 2. Let player 1 choose $x$ uniformly at random. For any guess $y$ by player 2, the number of matching bits $L$ follows a binomial distribution $B(\kappa, 1/2)$. The expected utility for player 2 is:
\[
\mathbb{E}[U_2] = \sum_{\ell=0}^\kappa 2^\ell \cdot \Pr[L=\ell] = \sum_{\ell=0}^\kappa 2^\ell \cdot \binom{\kappa}{\ell} \left(\frac{1}{2}\right)^\kappa = \left(\frac{1}{2}\right)^\kappa \sum_{\ell=0}^\kappa \binom{\kappa}{\ell} 2^\ell 1^{\kappa-\ell}
\]

By the Binomial Theorem, this simplifies to:
\[
\mathbb{E}[U_2] = \left(\frac{1}{2}\right)^\kappa (2+1)^\kappa = \left(\frac{3}{2}\right)^\kappa
\]

Since player 1 is uniform, this expectation is identical for any string $y$ that player 2 chooses. Thus, no unilateral deviation by player 2 can increase her expected utility, making the uniform profile a Nash equilibrium.
A similar argument shows that player~1 has no profitable unilateral deviation.
Suppose player~2 selects $y$ uniformly at random.
Then, for any (possibly deviating) choice of $x\in\{0,1\}^\kappa$ by player~1,
each bit of $y$ matches the corresponding bit of $x$ independently with probability $1/2$,
and therefore the number of matching bits $L$ is distributed as $L\sim \mathrm{Bin}(\kappa,1/2)$.
Hence player~1's expected utility is
\[
\mathbb{E}[U_1] \;=\; -\mathbb{E}[2^L] \;=\; -\left(\frac{3}{2}\right)^\kappa,
\]
which is independent of $x$. Thus, no unilateral deviation by player~1 can improve her expected utility,
and the uniform strategy profile is a Nash equilibrium in the ideal game.

Furthermore, because the utility distributions for any guess strategy are identical in $\idealGame$ (given the other player plays uniformly at random), they are trivially computationally indistinguishable; thus, they satisfy bidirectional computational mean dominance by \Cref{lemma: comp ind implies bidirectional dominance}. Thus, the uniform profile is also a pseudo-Nash equilibrium in the ideal game.

In the real game, it is interesting to see that the uniform strategy (Nash of the ideal game) is not always a computational Nash (\Cref{def: computational Nash}). We first notice that the real game is a computable sequence of games, on which the computational Nash is defined.  This is by definition (\Cref{def: hp computational game}): the real game always involves two players, has two rounds, and the utility can be efficiently computed. For a strategy to be a computational Nash, for all players, there is no unilateral deviation that can increase the expected utility by a noticeable amount. To prove the uniform strategy is a computational Nash, for every PPT unilateral deviation and its corresponding utility random variable, denoted by $U_2^*$, there must exist a negligible function $\epsilon$ such that
\begin{align*}
    \E{}{U_2} \ge \E{}{U_2^*}- \epsilon(\secparam),
\end{align*}

Consider a PPT deviation $s_2^*$ where player 2 attempts to brute-force the cryptographic commitment. By the hiding property, any PPT strategy has at most a negligible advantage $\mu(\kappa)$ in guessing the string $x$. However, if player 2 guesses correctly, she receives $2^\kappa$. Her expected utility becomes approximately:
\[
\mathbb{E}[U_2^*] \approx \mathbb{E}[U_2] + 2^\kappa \cdot \mu(\kappa)
\]
For a strategy to be CNE, it must hold that $\mathbb{E}[U_2] \geq \mathbb{E}[U_2^*] - \mathsf{negl}(\kappa)$. This requires $2^\kappa \cdot \mu(\kappa)$ to be negligible. Because standard cryptographic security only guarantees that $\mu(\kappa)$ is smaller than any inverse polynomial (and not necessarily $1/2^\kappa$), the ``utility spike'' $2^\kappa$ could make the expected gain non-negligible. Thus, we conclude that the uniform strategy is not necessarily a computational Nash equilibrium. Furthermore, we note that it is not easy to find a strategy profile that is a computational Nash, since this requires analyzing the exact probability of success, and cryptography does not provide such tools, but rather asymptotic ‘negligible’ upper bounds.

In contrast, the pseudo-Nash equilibrium remains stable because it is based on the empirical mean of $\secparam^c$ samples. That is because negligible-probability events (successful brute-force) are unlikely to be observed in any polynomial number of samples. Hence, the uniform strategy remains a pseudo-Nash equilibrium in the real game. This follows directly from the~\Cref{thr: pseudo implies pseudo}.

\subsection{Auctions with information advantage}\label{apx: auction}

In more detail, we analyze~\Cref{example: auction where ideal vs real is weird},  which illustrates the robustness of the pseudo-Nash equilibrium in settings where cryptographic implementation details create "spikes" in expected utility that would otherwise derail standard Nash equilibria.

We consider a single-item second-price auction with $n$ bidders. The bidders have values $\values=(v_1,\ldots,v_n)$ such that $v_i\in[0,V]$ for all $i$. Each bidder $i$ submits a sealed bid $b_i$ (potentially different from $v_i$). The auction rule is defined by $(x,p)$, where $x_i(\bids)$ is the allocation and $p_i(\bids)$ is the payment. In a second-price auction, the highest bidder wins and pays the second-highest bid. 

In addition to standard quasi-linear utility (i.e., the utility of agent $i$ is $u_i(\bids) = v_i\cdot x_i(\bids)-p_i(\bids)$), players gain utility from identifying the highest-valued bidder. This models competitive scenarios where knowing the identity of a market leader grants a future strategic advantage. Let $i^*=\argmax_i\{v_i\}$ be the identity of the highest-valued bidder. Let $\beta_{i,j}$ be player $i$’s posterior probability that player $j$ is the highest-valued bidder. Player i receives a reward $R(\secparam)$ with probability:
\[q_i = \max \left\{ \beta_{i,i^*} - \frac{1}{n-1}, 0 \right\}\]
Therefore, the utility for player $i$ is $U_i(\bids)=(v_i\cdot x_i(\bids) - p_i(\bids)) + R(\secparam)\cdot q_i.$

Given the players' utilities, we consider a standard sealed-bid second-price auction with an anonymity twist: bids are committed and then published in permuted order. Formally, we define the normal-form computational game:
\begin{framed}
\noindent \underline{Round 1 (Commit)}:
\begin{itemize}
    \item Each bidder $i$ picks a bid $b_i$ and randomness $r_i$, and sends $c_i=\Com{b_i;r_i}$ to the auctioneer.
    \item The auctioneer samples a permutation $\pi$ (ideal: uniform; real: $\pi=\PRP_s$ for $s\leftarrow\{0,1\}^\kappa$),
    and broadcasts the permuted list $C=(c_{\pi^{-1}(1)},\ldots,c_{\pi^{-1}(n)})$.
\end{itemize}

\noindent \underline{Round 2 (Open to auctioneer)}:
\begin{itemize}
    \item Each bidder $i$ sends an opening $(b_i,r_i)$ to the auctioneer, who verifies $\Com{b_i;r_i}=c_i$.
\end{itemize}

\noindent \underline{Auction and public price opening}:
\begin{itemize}
    \item The auctioneer computes the winner $i_{(1)}$ and price $p=b_{i_{(2)}}$ (second-highest bid).
    \item The auctioneer privately sends $p$ to the winner and $\bot$ to all other bidders.
    \item The auctioneer publicly opens \emph{only} the commitment at position $\pi(i_{(2)})$ in $C$,
    revealing $p$ (and the corresponding opening), but not the underlying identity $i_{(2)}$.
\end{itemize}
\end{framed}
For simplicity, one may treat $\Com\cdot$ as an ideal commitment functionality (perfect hiding/binding). In either the ideal or real commitment model, bidder $i$ can typically recognize her own commitment in the broadcast list (equivalently, learns $\pi(i)$), since she knows $c_i$. This does not help deanonymize other bidders’ commitments, and all conclusions below hinge only on whether the opened position $\pi(i_{(2)})$ can be linked back to an identity.

\textbf{Ideal game $\idealGame$.} First, we analyze the ideal game where the auctioneer uses a uniform random permutation $\pi$. Since the only publicly opened commitment corresponds to the price bidder $i_{(2)}$, and the opened position is uniformly permuted and independent of bidder identities, the permutation and subsequent opening provide no additional advantage for identifying $i^*$ beyond random guessing.
Hence, for any bidder $i$, the optimal posterior based solely on this permutation signal yields $\beta_{i,i^*}=1/(n-1)$ and thus $q_i=0$. In this case, the truthful strategy profile $\sigma$ (where $b_i=v_i$) is a Nash equilibrium, following directly from the truthfulness of the second-bid auction. We note that although the reward $R(\secparam)$ depends on the security parameter, the ideal game $\idealGame$ remains non-parameterized in the sense of \Cref{thrm: Nash equiv Pseudo}. That is because the perfect anonymity of the ideal permutation ensures that $q_i=0$ regardless of $\secparam$, causing the term $R(\secparam)\cdot q_i=0.$ Therefore, for every strategy profile $\tau$, the utility random variable $  \{U_i^\secparam(\tau)\}_{\secparam\in \mathbb{N}}$ is constant for all $\secparam$ (i.e  $U_i^\secparam(\tau)=u_i(\tau)$ for all $\secparam$), and the ideal-world game family collapses to a normal-form game with bounded utilities. Since truthful bidding is a Nash equilibrium in this game, \Cref{thrm: Nash equiv Pseudo} implies that $\sigma$ is also a pseudo-Nash equilibrium in $\idealGame$.

\textbf{Real game $\realGame$.} In the implementation of $\realGame,$ the uniform permutation is replaced by a pseudorandom permutation ($\PRP$). Before considering deviations, note that if all bidders follow the truthful bidding strategy, then, as in the ideal game, the only publicly opened commitment corresponds to the price (the second-highest bid). Hence, simply attempting to deanonymize (invert the pseudorandom permutation) the opened commitment does not help identify $i^*$, since it belongs to the second-highest bidder rather than the highest-valued one.

However, there is a profitable deviation for a large enough reward. Fix any PPT ``deanonymization'' algorithm $\Adv$ that, given the public transcript and the opened position $\pi(i_{(2)})$, outputs a guess for the underlying identity. By $\PRP$ security, $\Adv$'s advantage in identifying the correct preimage is bounded by a negligible function; concretely, there exists $\varepsilon(\secparam)=\negl{\secparam}$ such that whenever $\Adv$ is applied to a transcript in which the opened commitment corresponds to some target bidder $j$, we have
\[
\Pr[\Adv\text{ outputs }j] \;\le\; \tfrac{1}{n-1} + \varepsilon(\secparam).
\]
Now fix a bidder $i\neq i^*$ and consider the following unilateral deviation:
\begin{itemize}
    \item bidder $i$ bids an extremely large value $b_i\gg V$ to \emph{force winning} the auction.
    \item after the public opening, bidder $i$ runs $\Adv$ and sets her posterior $\beta_i$ to put mass $1$ on $\Adv$'s output (so her advantage over random guessing is on the order of $\varepsilon(\secparam)$ whenever $\Adv$ succeeds).
\end{itemize}

Under this deviation (and assuming other bidders bid truthfully), bidder $i$ wins, and the price becomes the highest bid among the remaining bidders, which is $b_{i^*}=v_{i^*}$. Hence, the publicly opened commitment belongs to $i^*$. This is precisely the situation in which deanonymizing the opened commitment reveals $i^*$.  The deviation changes bidder $i$'s auction payoff from $0$ (when losing truthfully) to $v_i - v_{i^*}$ (when forcing a win), i.e., it incurs auction loss at most $V$ in magnitude. On the other hand, it increases bidder $i$'s information term by an amount on the order of $R(\secparam)\cdot \varepsilon(\secparam).$ Thus, for sufficiently large $R(\secparam)$ (e.g., $R(\secparam)\cdot \varepsilon(\secparam) > v_{i^*}-v_i$), this deviation strictly increases bidder $i$'s expected utility, implying that the truthful profile is not a Nash equilibrium in the real game.

On the other hand, the truthful profile remains a pseudo-Nash equilibrium. Intuitively, the only gain from the overbidding deviation comes from a negligible advantage in deanonymizing the opened position, which is ``invisible'' under any polynomial number of samples. Formally, viewing the ideal uniform permutation as an ideal functionality $\F_{perm}$ and the $\PRP$-based permutation as a protocol $\Pi$ that securely realizes it, \Cref{thr: pseudo implies pseudo} implies that the truthful pseudo-Nash equilibrium in $\idealGame$ is preserved as a pseudo-Nash equilibrium in $\realGame$.

\subsection{HT game}\label{section:HTGAME}
The classical \textit{$t$-out-of-$n$ secret sharing} problem~\cite{Shamir79,Blakley79} involves a \textit{dealer} $D$ who wants to distribute a secret $s$ among a group of $n$ players, $P_1, P_2, \dots, P_n$. The scheme ensures that: (1) any subset of at least $t$ players can collaboratively reconstruct the secret without requiring further input from the dealer, while (2) any subset of fewer than $t$ players gains no information about the secret. One of the famous examples, introduced by Shamir~\cite{Shamir79}, works as follows: the secret $m^*$ is an element of a finite field $\mathbb{F}$, where $|\mathbb{F}| > n$. The dealer selects a random polynomial $f(x)$ of degree at most $t-1$ such that $f(0) = m^*$. Each player $P_i$ receives a share $s_i = f(i)$. One can check that a subset of at least $t$ players can determine $f(x)$—and thereby recover $m^*$—through polynomial interpolation. Conversely, any subset of fewer than $t$ players learns nothing about $m^*$ due to the random choice of the polynomial $f$.

In the above classical setting, there is an assumption that at least $t\geq 2n/3$ players are willing to collaborate and pool their shares when reconstructing the secret. 
This 
assumption 
is known to be necessary and sufficient for secret sharing with robust reconstruction against malicious (arbitrarily deviating) parties. 

Later, Halpern and Teague \cite{halpern2004rational} explore a setting where players are neither fully honest nor completely malicious but are instead \textit{rational}, a concept known as \textit{$t$-out-of-$n$ rational secret sharing}. Depending on the utility functions of the players, one can observe that Shamir's protocol is not a Nash equilibrium in the rational setting~\cite{halpern2004rational}. Specifically, each player prefers learning the secret above all else, however, if given a choice, prefers that as few others as possible also learn the secret (an example utility is given in~\Cref{table: ht game}). Moreover, Halpern and Teague~\cite{halpern2004rational} showed that, beyond Shamir's protocol, any protocol stopping in a fixed number of rounds cannot be a Nash equilibrium. 

In the model of \textit{rational secret sharing}~\cite{halpern2004rational}, which we refer to as the \textit{HT game} in this section, there is a \textit{dealer} $D$ holding a secret $m^*$ and $n$ players, $P_1, \dots, P_n$. A threshold $t$ is fixed at the outset and known to all players. For our discussion, we focus on the specific case where $n=3$ and $t=3$. The protocol proceeds through multiple \textit{communication rounds}. At the start of the HT game, $D$ privately distributes information to each of the three players, ensuring that no subset of players gains any information about the secret $m^*$. After this initial distribution, the dealer no longer participates. Instead, the three players, all assumed to be rational, execute the protocol by simultaneously broadcasting messages over multiple rounds. We defer a detailed description of the HT game to~\Cref{table: ht game}, which can be found in~\Cref{app: def and protocols for the HT game}.

The \textit{Halpern-Teague protocol} in the 3-out-of-3 case, denoted by $s^*$ (a detailed description of $s^*$ can be found in~\Cref{strateg: s_star for ht game}, \Cref{app: def and protocols for the HT game}), proceeds in a sequence of \textit{iterations}, with each iteration consisting of \textit{four communication rounds}. During the $j$-th iteration, each player $P_i$ flips a biased coin $dc^j_i$, which takes the value $1$ with probability $\alpha$. The players then jointly compute an \textit{information-theoretically secure computation} to compute the value $p^j = dc^j_1 \oplus dc^j_2 \oplus dc^j_3$. Notably, it is impossible for any player to cheat (except by aborting the protocol) or to gain information about the values $dc^j_i$ of other players beyond what is implied by $p^j$. If $p^j = dc^j_i = 1$, player $P_i$ broadcasts their share. If all shares are revealed, the secret is reconstructed, and the game terminates. However, if $p^j = 1$ and either no shares or exactly two shares are revealed, or if the secure computation of $p^j$ is aborted, the game ends immediately. In all other cases, the players proceed to the next iteration.

Note that in the Halpern-Teague protocol $s^*$, the secret is reconstructed only if $dc^j_1 = dc^j_2 = dc^j_3 = 1$ in some iteration $j$. Thus, assuming all players follow the protocol $s^*$, the protocol terminates in each iteration with probability $\alpha^3$. As a result, the number of iterations required follows a geometric distribution with parameter $\alpha^3$. This implies that the Halpern-Teague protocol $s^*$ does not have a fixed upper bound on its round complexity. \Cref{coro: s* is nash for HT game} establishes that $s^*$ is a Nash equilibrium.

\begin{corollary}
    \label{coro: s* is nash for HT game}
    For every $\secparam \geq 1,$ the strategy $s^*(1^\secparam)$ (shown in~\Cref{strateg: s_star for ht game}) is a Nash equilibrium of the HT game, shown in~\Cref{table: ht game}.
\end{corollary}

\begin{proof}
    \label{proof: coro: s* is nash for HT game}
    We start with a few definitions. Let $i \in \{1, 2, 3\}$ be the index of an arbitrary fixed player in the game. Let $\tilde{s}_i$ be an arbitrary fixed strategy for player $i$, and let $\secparam \in \mathbb{N}$ be an arbitrary fixed integer. For brevity, we will denote $s^*(1^\secparam)$ by $s^*$ whenever the context of the proof is clear. Recall $s^* = (s_i^*, s^*_{-i})$ is the strategy profile defined in the statement. Denote by $\tilde{s} = (\tilde{s}_i, s^*_{-i})$ the unilateral deviated strategy. Denote by $U_i(s^*)$ and $U_i(\tilde{s})$ the utility random variable of player $i$ when players playing the strategy profiles $s^*$ and $\tilde{s}$, accordingly. 
    
    To show that $s^*$ is a Nash, it suffices to show the expected utility of player $i$ when players playing the strategy $s^*$ is greater than or equal to the expected utility of player $i$ when players playing $\tilde{s}$, or formally
    \begin{align*}
        \Ex{U_i(s^*)} \geq \Ex{U_i(\tilde{s})}.
    \end{align*}
    
    Note that the strategy $\tilde{s}_i$ can only deviate by choosing to stop and hiding its secret share in some round $T_0$. Let random variable $\mathcal{T}$ denote the number of round until game stops when players playing strategy profile $s^*.$ One can check that 
    \begin{align*}
        \Ex{U_i(s^*) \mid \mathcal{T}\leq T_0} = \Ex{U_i(\tilde{s}) \mid \mathcal{T}\leq T_0}.
    \end{align*}

    In contrast, we still have $\Ex{U_i(s^*) \mid \mathcal{T}\leq T_0} = 2^{\secparam}$, conditioned on the event $\mathcal{T} > T_0$, since the game when playing $s^*$ will continue until every one learns the secret $m^*.$ However, when playing $\tilde{s}$, the game will stops exactly at round $T_0$. Since we have $\mathcal{T} > T_0$, and the other two players follow $s^*$, player $i$ will have only two outcomes: either only $i$ learns the secret, or no one, which gives 
    \begin{align*}
        \Ex{U_i(\tilde{s}) \mid \mathcal{T}> T_0} = \frac{\alpha^2}{\alpha^2 + (1-\alpha)^2} (2^{k+1}+2) + \frac{(1-\alpha)^2}{\alpha^2 + (1-\alpha)^2}.
    \end{align*}
    One can check that when plugging $\alpha = \frac{1}{3}$, for all $\secparam \geq 1,$
    \begin{align*}
        \Ex{U_i(s^*) \mid \mathcal{T}> T_0} \geq \Ex{U_i(\tilde{s}) \mid \mathcal{T}> T_0}.
    \end{align*}

    \noindent Combining both cases, we conclude the proof by observing $\Ex{U_i(s^*)} \geq \Ex{U_i(\tilde{s})}.$ 
    
\end{proof}

We consider a finite-round variant of the Halpern-Teague protocol, denoted as the stopping strategy $\hat{s}$ (a detailed description of $\hat{s}$ can be found in~\Cref{strateg: s_hat for ht game}, \Cref{app: def and protocols for the HT game}). $\hat{s}$ follows exactly the strategy given by $s^*$ but enforces termination after a fixed number of rounds, which is linear in the size of the share used in the game. Interestingly, while the Halpern-Teague protocol $s^*$ terminates with high probability in a finite number of rounds --- determined by the parameter $\alpha^3$ --- any strategy, e.g., $\hat{s}$, that enforces termination in a fixed number of rounds is ruled out as a stable equilibrium. This apparent contradiction arises because such a finite-round stopping strategy fails to satisfy the conditions of a Nash equilibrium.  However, the authors prove that there exists an infinite-rounds strategy which is Nash. 

What is interesting here is that the proposed Nash has an overwhelming probability of being a finite round strategy. In fact, it is not hard to find a fixed bound such that this strategy will go beyond this bound only if a negligible probability event occurs. This is exactly an artifact that pseudo-Nash was designed to take care of: As we show in~\Cref{thm: s_hat is a pseudo for HT game}, the above (finite) stopping strategy $\hat{s}$ is  \textit{pseudo-Nash}, which resolves this counter-intuitive situation. 

\begin{theorem}
    \label{thm: s_hat is a pseudo for HT game}
    The strategy $\hat{s}$ (shown in~\Cref{strateg: s_hat for ht game}) is a pseudo-equilibrium of the HT game, shown in~\Cref{table: ht game}.
\end{theorem}

\begin{proof}
    \label{proof: thm: s_hat is a pseudo for HT game}
    We start with a few definitions. Let $i \in \{1, 2, 3\}$ be the index of an arbitrary fixed player in the game. Let $c \in \mathbb{N}$ be an arbitrary fixed integer. Let $\tilde{s}_i$ be an arbitrary fixed strategy for player $i$. For brevity, we will denote $\hat{s}(1^\secparam)$ by $\hat{s}$ whenever the context of the proof is clear. Recall $\hat{s} = (\hat{s}_i, \hat{s}_{-i})$ is the strategy profile defined in the statement. Denote by $\tilde{s} = (\tilde{s}_i, \hat{s})$ the unilateral deviated strategy. Denote by $U_i(\hat{s})$ and $U_i(\tilde{s})$ the utility random variable of player $i$ when players playing the strategy profiles $\hat{s}$ and $\tilde{s}$, accordingly. 

    To show that $\hat{s}$ is a pseudo-equilibrium, it suffices to show there exists a $\secparam_1$ such that for all $\secparam \geq \secparam_1$:

    \begin{align*}
\begin{aligned}
&\meanDominanceEQ{U_i(\tilde{s})}{U_i(\hat{s})}{\secparam^{4c}}
\\
&\qquad\qquad\qquad\qquad\qquad-\meanDominanceEQ{U_i(\hat{s})}{U_i(\tilde{s})}{\secparam^{4c}}
< \frac{1}{\secparam^{c}} \, .
\end{aligned}
\end{align*}

    where $U_i(\hat{s})^{(j)},U_i(\tilde{s})^{(j)}$ are i.i.d. samples from $U_i(\hat{s}),U_i(\tilde{s})$ respectively.

    Let $\mathcal{S}_{0, \secparam}$ be the event that the game, when played under the strategy profile (also referred to as the protocol) $\hat{s}$, ends within $4\secparam$ rounds. Following a similar argument of \Cref{coro: s* is nash for HT game}, one can check
    \begin{align*}
        \Ex{U_i(\hat{s}) \mid \mathcal{S}_{0, \secparam}} \geq \Ex{U_i(\tilde{s}) \mid \mathcal{S}_{0, \secparam}}.
    \end{align*}
    Let $\rZ_\secparam$ be the random variable denoting number of rounds until the protocol $\hat{s}$ terminates. Let $\rZ^{(1)}_\secparam, \cdots, \rZ^{(\secparam^{4c})}_\secparam$ be $\secparam^{4c}$ independent copies of $\rZ_\secparam$. Let $\mathcal{S}_{1, \secparam}$ be the event that all $\secparam^{4c}$ independent runs of the game when playing the strategy profile $\hat{s}$ ends in $4\secparam$ rounds, i.e., $\rZ^{(j)}_\secparam \leq 4\secparam$ holds for all $j \in [\secparam^{4c}]$. One can check that the event $\mathcal{S}_{1, \secparam}$ occurs with overwhelming probability. Formally, there exists a negligible function $\delta(\cdot)$ such that $\pr{\mathcal{S}_{1, \secparam}} \geq 1 - \delta(\secparam)$. This is because 
    \begin{align*}
           &\pr{\mathcal{S}_{1, \secparam}}\\
        ={}& \pr{\rZ_\secparam \leq 4\secparam}^{\secparam^{4c}} \tag{$\rZ^{(1)}_\secparam, \cdots, \rZ^{(\secparam^{4c})}_\secparam$ are independent}\\
        ={}&\left(1 - (1-\alpha^3)^{\secparam} \right)^{\secparam^{4c}}  \tag{$\pr{\rZ_\secparam \leq 4\secparam} = 1 - (1-\alpha^3)^{\secparam}$}\\
        \geq{}& 1 - \secparam^{4c}(1-\alpha^3)^{\secparam} \tag{By Bernoulli’s inequality}\\
        ={}& 1 - \frac{\secparam^{4c}26^\secparam}{27^\secparam} \tag{$\alpha = \frac{1}{3}$}\\
        ={}& 1 - \delta(\secparam).
    \end{align*}

    Lastly, we conclude the proof by observing that
    \begin{align*}
        & \meanDominanceEQ{U_i(\tilde{s})}{U_i(\hat{s})}{\secparam^{4c}}-\meanDominanceEQ{U_i(\hat{s})}{U_i(\tilde{s})}{\secparam^{4c}}\\
        ={}& \meanDominance{U_i(\tilde{s})}{U_i(\hat{s})}{\secparam^{4c}}-\meanDominance{U_i(\hat{s})}{U_i(\tilde{s})}{\secparam^{4c}}\\
        \leq{}& \conMeanDominance{U_i(\tilde{s})}{U_i(\hat{s})}{\secparam^{4c}}{\mathcal{S}_{1, \secparam}}\pr{\mathcal{S}_{1, \secparam}} + \delta(\secparam)\\
        &\qquad\qquad\qquad-\meanDominance{U_i(\hat{s})}{U_i(\tilde{s})}{\secparam^{4c}}\\
        \leq{}& \conMeanDominance{U_i(\tilde{s})}{U_i(\hat{s})}{\secparam^{4c}}{\mathcal{S}_{1, \secparam}} + \delta(\secparam)\\
        &\qquad\qquad\qquad-\meanDominance{U_i(\hat{s})}{U_i(\tilde{s})}{\secparam^{4c}}\\
        \leq{}& \conMeanDominance{U_i(\tilde{s})}{U_i(\hat{s})}{\secparam^{4c}}{\mathcal{S}_{1, \secparam}} + \delta(\secparam) \\
        &\qquad\qquad\qquad-\conMeanDominance{U_i(\hat{s})}{U_i(\tilde{s})}{\secparam^{4c}}{\mathcal{S}_{1, \secparam}} \pr{{\mathcal{S}_{1, \secparam}}}\\
        \leq{}& \conMeanDominance{U_i(\tilde{s})}{U_i(\hat{s})}{\secparam^{4c}}{\mathcal{S}_{1, \secparam}} + 2\delta(\secparam)\\
        &\qquad\qquad\qquad-\conMeanDominance{U_i(\hat{s})}{U_i(\tilde{s})}{\secparam^{4c}}{\mathcal{S}_{1, \secparam}}\\
        \leq{}& 2\delta(\secparam). \tag{$\Ex{U_i(\hat{s}) \mid \mathcal{S}_{0, \secparam}} \geq \Ex{U_i(\tilde{s}) \mid \mathcal{S}_{0, \secparam}}$}
    \end{align*}
\end{proof}

The HT game and the stopping strategy $\hat{s}$ also well illustrate the separation between the \textit{pseudo-equilibrium} notion and approximate equilibrium concepts, both from classical game theory, e.g., \textit{$\epsilon$-Nash equilibrium} and from past attempts to crypto-friendly definitions, e.g., \textit{Computational Nash equilibrium}~\cite{halpern2016computational}. Informally, a strategy profile in a computational game is a \textit{computational Nash equilibrium} if no polynomial-time unilateral deviation results in a noticeably higher expected utility. Unlike $\epsilon$-Nash and computational Nash equilibria, which are sensitive to the magnitude of negligible probability events (particularly in our HT game, where the utility can be exponentially large), the pseudo-equilibrium notion remains robust against such anomalies. This robustness allows us to resolve counter-intuitive impossibility results, as the stopping strategy $\hat{s}$ fails to be an $\epsilon$-Nash equilibrium or a computational Nash equilibrium. The formal description of these impossibility results is stated in \Cref{thm: s_hat is not a eps-nash for HT game} and~\Cref{coro: s_hat is not a computational-nash for HT game}.

\begin{theorem}
    \label{thm: s_hat is not a eps-nash for HT game}
    For every constant $\epsilon > 0$, and for all $\secparam > \max\{2, \lceil \frac{\ln{\epsilon}}{\ln{52/27}} \rceil\}$, the strategy $\hat{s}(1^\secparam)$ (shown in~\Cref{strateg: s_hat for ht game}) is not a $\epsilon$-Nash equilibrium of the HT game, shown in~\Cref{table: ht game}.
\end{theorem}

\begin{proof}
    \label{proof: thm: s_hat is not a eps-nash for HT game}
    We start with a few definitions. Let $i \in \{1, 2, 3\}$ be the index of an arbitrary fixed player in the game. For brevity, we denote the strategy profiles $s^*(1^\secparam)$ and $\hat{s}(1^\secparam)$ by $s^*$ and $\hat{s}$, respectively, whenever the context of the proof is clear. Denote by $U_i(s^*)$ and $U_i(\hat{s})$ the utility random variable of player $i$ when players playing the strategy profiles $s^*$ and $\hat{s}$, accordingly. 
    
    To prove the statement, it suffices to show that there exists a positive integer $\secparam_1 \in \mathbb{N}$ such that for every $\secparam \geq \secparam_1$, the following holds
    \begin{align*}
        \Ex{U_i(\hat{s})} < \Ex{U_i(s^*)} - \epsilon.
    \end{align*}
    
    Let $\rZ_\secparam$ be the random variable denoting number of rounds until the protocol $s^{*}$ terminates. We have 
    \begin{align*}
           & \Ex{U_i(s^*)} - \Ex{U_i(\hat{s})} \\
        ={}& \Ex{U_i(s^*) \mid \rZ_\secparam > 4\secparam}\pr{\rZ_\secparam > 4\secparam} + \Ex{U_i(s^*) \mid \rZ_\secparam \leq 4\secparam}\pr{\rZ_\secparam \leq 4\secparam} - \Ex{U_i(\hat{s})}\\
        ={}& \Ex{U_i(s^*) \mid \rZ_\secparam > 4\secparam}\pr{\rZ_\secparam > 4\secparam} + \Ex{U_i(\hat{s}) \mid \rZ_\secparam \leq 4\secparam}\pr{\rZ_\secparam \leq 4\secparam} - \Ex{U_i(\hat{s})}\\
        ={}& \left(\Ex{U_i(s^*) \mid \rZ_\secparam > 4\secparam} - \Ex{U_i(\hat{s}) \mid \rZ_\secparam > 4\secparam}\right)\pr{\rZ_\secparam > 4\secparam}\\
        ={}& \left(2^{\secparam} -1 \right)\pr{\rZ_\secparam > 4\secparam}\\
        ={}& \left(2^{\secparam} -1 \right) (1-\alpha^3)^{\secparam} \tag{$\pr{\rZ_\secparam \leq 4\secparam} = 1 - (1-\alpha^3)^{\secparam}$} \\
        ={}& \left(2^{\secparam} -1 \right) \left(\frac{26}{27}\right)^{\secparam} \tag{$\alpha = \frac{1}{3}$}.
    \end{align*}

    Setting $\secparam_1 = \max\{2, \lceil \frac{\ln{\epsilon}}{\ln{52/27}} \rceil\}$, where $\lceil \cdot \rceil$ is the ceiling function, one can check for every $\secparam > \secparam_1$,
    \begin{align*}
        \Ex{U_i(s^*)} > \Ex{U_i(\hat{s})} + \epsilon.
    \end{align*}
    Therefore, the strategy $\hat{s}(1^\secparam)$ is not a $\epsilon$-Nash equilibrium of the HT game.
\end{proof}

\begin{corollary}
    \label{coro: s_hat is not a computational-nash for HT game}
    The strategy $\hat{s}$ (shown in~\Cref{strateg: s_hat for ht game}) is not a computational Nash equilibrium of the HT game, shown in~\Cref{table: ht game}.
\end{corollary}

\begin{proof}
    \label{proof: coro: s_hat is not a computational-nash for HT game}
    We start by defining a protocol $\tilde{s}$ for the HT game. $\tilde{s}$ follows exactly the stopping strategy given by $\hat{s}$, but enforces termination after $8 \secparam$ rounds, instead of the $4 \secparam$ rounds enforced in $\hat{s}$. One can check that both protocols $\tilde{s}$ and $\hat{s}$ run in polynomial-time (in terms of the parameter $\secparam$).
    
    To prove the statement, it suffices to show that there exists a positive integer $\secparam_1 \in \mathbb{N}$ and a polynomial function $\poly{\cdot}$ such that for every $\secparam \geq \secparam_1$, the following holds
    \begin{align*}
        \Ex{U_i(\tilde{s})} - \Ex{U_i(\hat{s})} > \frac{1}{\poly{\secparam}}.
    \end{align*}

    Let $\rZ_\secparam$ be the random variable denoting number of rounds until the protocol $\tilde{s}$ terminates. One can check 
    \begin{align*}
           & \Ex{U_i(\tilde{s})} - \Ex{U_i(\hat{s})} \\
        ={}& \Ex{U_i(\tilde{s}) \mid \rZ_\secparam > 4\secparam}\pr{\rZ_\secparam > 4\secparam} + \Ex{U_i(\tilde{s}) \mid \rZ_\secparam \leq 4\secparam}\pr{\rZ_\secparam \leq 4\secparam} - \Ex{U_i(\hat{s})}\\
        ={}& \Ex{U_i(\tilde{s}) \mid \rZ_\secparam > 4\secparam}\pr{\rZ_\secparam > 4\secparam} + \Ex{U_i(\hat{s}) \mid \rZ_\secparam \leq 4\secparam}\pr{\rZ_\secparam \leq 4\secparam} - \Ex{U_i(\hat{s})}\\
        ={}& \left(\Ex{U_i(\tilde{s}) \mid \rZ_\secparam > 4\secparam} - \Ex{U_i(\hat{s}) \mid \rZ_\secparam > 4\secparam}\right)\pr{\rZ_\secparam > 4\secparam}\\
        ={}& \left(\Ex{U_i(\tilde{s}) \mid \rZ_\secparam > 4\secparam} - 1\right)\pr{\rZ_\secparam > 4\secparam}\\
        ={}& \Ex{U_i(\tilde{s}) \mid \rZ_\secparam > 8\secparam}\pr{\rZ_\secparam > 8\secparam} + \Ex{U_i(\tilde{s}) \mid 4\secparam \leq \rZ_\secparam \leq  8\secparam}\pr{4\secparam \leq \rZ_\secparam \leq 8\secparam}\\ 
        &\qquad- \pr{\rZ_\secparam > 4\secparam}\\
        ={}& \left(2^{\secparam} -1 \right)\pr{4\secparam \leq \rZ_\secparam \leq 8\secparam}\\
        ={}& \left(2^{\secparam} -1 \right) \left((1-\alpha^3)^{\secparam} - (1-\alpha^3)^{2\secparam} \right)\\
        ={}& \left(2^{\secparam} -1 \right) \left(\left(\frac{26}{27}\right)^{\secparam}\left(1 - \left(\frac{26}{27}\right)^{\secparam}\right) \right),
    \end{align*}
    which is not a negligible function.
\end{proof}

\subsection{HT Game: Definition and Protocols}
\label{app: def and protocols for the HT game}
\begin{table}[H]
\begin{tabularx}{\textwidth}{|X|}
\hline\\
\textbf{HT Game} \\
\hline
\begin{itemize}
    \item \textbf{Notations.} There are three players $P_1, P_2, P_3$.  For a player index $i \in \{1, 2, 3\}$, let $i^{+}$ denote $i+1$, except that $3^{+}$ is $1$. Similarly, let $i^{-}$ denote $i-1$, except that $1^{-}$ wraps around to $3$. For ease of understanding, and whenever it is not ambiguous, we refer to Player $P_{i^{+}}$ as the \emph{right neighbor} of Player $P_i$, and Player $P_{i^{-}}$ as the \emph{left neighbor} of Player $P_i$.
    
    \item \textbf{Game Setup.} Each player $P_i$ holds an infinite sequence of shares $\{s_i^1, s_i^2, \dots\}$, where $s_i^j \in \bits^{\poly{\secparam}}$. For $j$-th element, where $j \in \mathbb{N}^+$, the tuple $\bigl(s_1^j,s_2^j,s_3^j\bigr)$ forms a 3-out-of-3 secret sharing of the secret $m^{*}$. Shares in the sequence $\{s_i^1, s_i^2, \dots\}$ are mutually independent, i.e., $s_i^j$ is independent of $s_i^k$ for all $j\neq k$.

    \item \textbf{Actions per Round.} In the round $t \in \{1, 2, \dots\}$, each player $P_i$ chooses an action tuple $(l_i^t, r_i^t)$, where $l_i^t$ and $r_i^t$ represent the actions of sending a message to the \emph{left neighbor} (Player $P_{i^{-}}$) and the \emph{right neighbor} (Player $P_{i^{+}}$), respectively. Each action $l_i^t, r_i^t \in \bits^{\poly{\secparam}} \cup \{\varnothing, \bot\}$ can be one of the following:
    \begin{itemize}
        \item $m_i^j \in \bits^{\poly{\secparam}}$: Send a polynomial-length message. 
        \item $\varnothing$: Send a null message.
        \item $\bot$: Abort the game immediately.
    \end{itemize}
    All action tuples $(l_i^t, r_i^t)$ are chosen simultaneously at the start of each round.
    
    \item \textbf{Outcome of a Round.} At the end of round $t$, each player $P_i$ observes the tuple $(r_{i^{-}}^t, l_{i^{+}}^t)$, where $r_{i^{-}}^t$ is the message received from its \emph{left neighbor} (Player $P_{i^{-}}$), and $l_{i^{+}}^t$ is the message received from its \emph{right neighbor} (Player $P_{i^{+}}$). Then:
    \begin{itemize}
        \item If at least one player can reconstruct $m^*$ from the messages it received and its own shares, (i.e., from $\{(r_{i^{-}}^j, l_{i^{+}}^j), s_i^j\}_{j \leq t}$,) the game ends.
        \item If at least one player chooses to abort ($\bot$), the game also ends. 
        \item Otherwise, the game proceeds to the next round.
    \end{itemize}

    \item \textbf{Utilities of the Game.} All players share the same utility function. Let $T$ denote the final round index. For each player $P_i$, define $\mathrm{info}_i = \{(r_{i^{-}}^j, l_{i^{+}}^j), s_i^j\}_{j \leq T}$ as the set of all information available to Player $P_i$, which includes all messages received from other players during the $T$ rounds and its own shares. Let $\mathrm{rec}_{m^*}: \bits^{*} \to \bits$ be the reconstruction function for the secret $m^*$. The function $\mathrm{rec}_{m^*}$ takes as input an arbitrary message and outputs $1$ if it can successfully reconstruct $m^*$; otherwise, it outputs $0$. Then, the utility of each player $P_i$, denoted as $u_i$, is determined as follows:
    \begin{itemize} 
        \item If only player $i$ learns the secret $m^*$, i.e., $\mathrm{rec}_{m^*}(\mathrm{info}_i) = 1$ and $\mathrm{rec}_{m^*}(\mathrm{info}_{i^{+}}) = \mathrm{rec}_{m^*}(\mathrm{info}_{i^{-}}) = 0$, then $u_i = 2^{\secparam + 1} + 2$.
        \item If all players learn the secret $m^*$, i.e., $\mathrm{rec}_{m^*}(\mathrm{info}_i) = \mathrm{rec}_{m^*}(\mathrm{info}_{i^{+}}) = \mathrm{rec}_{m^*}(\mathrm{info}_{i^{-}}) = 1$, then $u_i = 2^{\secparam}$.
        \item If no player learns the secret $m^*$, i.e., $\mathrm{rec}_{m^*}(\mathrm{info}_i) = \mathrm{rec}_{m^*}(\mathrm{info}_{i^{+}}) = \mathrm{rec}_{m^*}(\mathrm{info}_{i^{-}}) = 0$, then $u_i = 1$.
    \end{itemize}
\end{itemize} \\\hline
\end{tabularx}
\caption{Secret Sharing Game (HT Game)~\cite{halpern2004rational}}\label{table: ht game}
\end{table}

\begin{table}[H]
\begin{tabularx}{\textwidth}{|X|}
\hline\\
\textbf{Strategy $s^*$ for HT Game} \\
\hline
\begin{itemize}
    \item \textbf{Actions in Rounds of the Form $t = 4j + 1$.} For each round $t = 4j + 1$, where $j \in \{0,1,2,\dots\}$, every player $P_i$ does the following:
    \begin{itemize}
        \item Flip a biased coin $dc^j_i$ with parameter $\alpha = \frac{1}{3}$, i.e., $dc^j_i \sim \text{Bernoulli}(\alpha)$. We call $dc^j_i$ the $j$-th \emph{decision coin} of player $P_i$.
        \item Flip a fair coin $mc^j_i$, i.e., $mc^j_i \sim \text{Bernoulli}(\frac{1}{2})$. We call $mc^j_i$ the $j$-th \emph{mask} of player $P_i$.
              
        \item Send the masked decision coin to the \emph{left neighbor}, i.e., Set $l_{i}^t \leftarrow dc^j_i \oplus mc^j_i$.
        \item Send the mask to the \emph{right neighbor}, i.e., Set $r_{i}^t \leftarrow mc^j_i$.
    \end{itemize}
    
    \item \textbf{Actions in Rounds of the Form $t = 4j + 2$.} For each round $t = 4j + 2$, where $j \in \{0,1,2,\dots\}$, every player $P_i$ does the following:
    \begin{itemize}
        \item Send the masked double-decision coin $dc^j_i \oplus l^{t-1}_{i^+} = dc^j_i \oplus dc^j_{i^+} \oplus mc^j_{i^+}$ to the \emph{left neighbor}, where $dc^j_i$ is player $P_i$'s $j$-th decision coin, and $l^{t-1}_{i^+} = dc^j_{i^+} \oplus mc^j_{i^+}$ is the $j$-th masked decision coin received from its right neighbor $P_{i^+}$. That is, Set $l_{i}^t \leftarrow dc^j_i \oplus l^{t-1}_{i^+}$.
        
        \item Send a null message to the \emph{right neighbor},  i.e., Set $r_{i}^t \leftarrow \varnothing$.
    \end{itemize}

    \item \textbf{Actions in Rounds of the Form $t = 4j + 3$.} For each round $t = 4j + 3$, where $j \in \{0,1,2,\dots\}$, every player $P_i$ does the following:
    \begin{itemize}
        \item Compute the $j$-th group decision bit $p^j = dc^j_1 \oplus dc^j_2 \oplus dc^j_3$, where $dc^j_1, dc^j_2, dc^j_3$ are the $j$-th decision coin of player $P_1, P_2, P_3$. $p$ can be computed from its passed three round message $l_{i^{+}}^{t-1} \oplus r_{i^-}^{t-2} \oplus dc^j_i = dc^j_{i^+} \oplus dc^j_{i^-} \oplus mc^j_{i^-} \oplus mc^j_{i^-} \oplus dc^j_i = dc^j_1 \oplus dc^j_2 \oplus dc^j_3 = p^j$
        \item If both the group decision and its own decision for this round is to share, i.e., $p^j = dc^j_i = 1$, then set $r_{i}^t \leftarrow s^j_i$ and $l_{i}^t \leftarrow s^j_i$; otherwise set $r_{i}^t \leftarrow \varnothing$ and $l_{i}^t \leftarrow \varnothing$.
    \end{itemize}

    \item \textbf{Actions in Rounds of the Form $t = 4j + 4$.} For each round $t = 4j + 4$, where $j \in \{0,1,2,\dots\}$, every player $P_i$ does the following:
    \begin{itemize}
        \item If the $j$-th group decision is not to share and player $P_i$ receives no shares from others, i.e., $p^j = 0$ and $l_{i^+}^{t-1} = r_{i^-}^{t-1} = \varnothing$; OR
        \item both the $j$-th group decision and $P_i$'s decision are to share and player $P_i$ receives no shares from others, i.e., $p^j = dc^j_i = 1$ and $l_{i^+}^{t-1} = r_{i^-}^{t-1} = \varnothing$; OR
        \item the $j$-th group decision is to share, player $P_i$'s decision is not to share and player $P_i$ receives exactly one share from others, i.e., $p^j = 1, dc^j_i = 0$ and either $(l_{i^+}^{t-1} = \varnothing, r_{i^-}^{t-1} = s_{i^-}^j)$ or $(l_{i^+}^{t-1} = s_{i^+}^j, r_{i^-}^{t-1} = \varnothing)$, then player $P_i$ requests the game to continue by setting $r_{i}^t \leftarrow \varnothing$ and $l_{i}^t \leftarrow \varnothing$.
        \item Otherwise $P_i$ aborts the game by $r_{i}^t \leftarrow \bot$ and $l_{i}^t \leftarrow \bot$.
    \end{itemize}
\end{itemize} \\\hline
\end{tabularx}
\caption{Strategy $s^*$ for HT Game}\label{strateg: s_star for ht game}
\end{table}

\begin{table}[H]
\begin{tabularx}{\textwidth}{|X|}
\hline\\
\textbf{Strategy $\hat{s}$ for HT Game} \\
\hline
\begin{itemize}
    \item \textbf{Actions in Rounds of the Form $t = 4j + 1$.} For each round $t = 4j + 1$, where $j \in \{0,1,2,\dots, \secparam\}$, every player $P_i$ does the following:
    \begin{itemize}
        \item Flip a biased coin $dc^j_i$ with parameter $\alpha = \frac{1}{3}$, i.e., $dc^j_i \sim \text{Bernoulli}(\alpha)$. We call $dc^j_i$ the $j$-th \emph{decision coin} of player $P_i$.
        \item Flip a fair coin $mc^j_i$, i.e., $mc^j_i \sim \text{Bernoulli}(\frac{1}{2})$. We call $mc^j_i$ the $j$-th \emph{mask} of player $P_i$.
              
        \item Send the masked decision coin to the \emph{left neighbor}, i.e., Set $l_{i}^t \leftarrow dc^j_i \oplus mc^j_i$.
        \item Send the mask to the \emph{right neighbor}, i.e., Set $r_{i}^t \leftarrow mc^j_i$.
    \end{itemize}
    
    \item \textbf{Actions in Rounds of the Form $t = 4j + 2$.} For each round $t = 4j + 2$, where $j \in \{0,1,2,\dots, \secparam\}$, every player $P_i$ does the following:
    \begin{itemize}
        \item Send the masked double-decision coin $dc^j_i \oplus l^{t-1}_{i^+} = dc^j_i \oplus dc^j_{i^+} \oplus mc^j_{i^+}$ to the \emph{left neighbor}, where $dc^j_i$ is player $P_i$'s $j$-th decision coin, and $l^{t-1}_{i^+} = dc^j_{i^+} \oplus mc^j_{i^+}$ is the $j$-th masked decision coin received from its right neighbor $P_{i^+}$. That is, Set $l_{i}^t \leftarrow dc^j_i \oplus l^{t-1}_{i^+}$.
        
        \item Send a null message to the \emph{right neighbor},  i.e., Set $r_{i}^t \leftarrow \varnothing$.
    \end{itemize}

    \item \textbf{Actions in Rounds of the Form $t = 4j + 3$.} For each round $t = 4j + 3$, where $j \in \{0,1,2,\dots, \secparam\}$, every player $P_i$ does the following:
    \begin{itemize}
        \item Compute the $j$-th group decision bit $p^j = dc^j_1 \oplus dc^j_2 \oplus dc^j_3$, where $dc^j_1, dc^j_2, dc^j_3$ are the $j$-th decision coin of player $P_1, P_2, P_3$. $p$ can be computed from its passed three round message $l_{i^{+}}^{t-1} \oplus r_{i^-}^{t-2} \oplus dc^j_i = dc^j_{i^+} \oplus dc^j_{i^-} \oplus mc^j_{i^-} \oplus mc^j_{i^-} \oplus dc^j_i = dc^j_1 \oplus dc^j_2 \oplus dc^j_3 = p^j$
        \item If both the group decision and its own decision for this round is to share, i.e., $p^j = dc^j_i = 1$, then set $r_{i}^t \leftarrow s^j_i$ and $l_{i}^t \leftarrow s^j_i$; otherwise set $r_{i}^t \leftarrow \varnothing$ and $l_{i}^t \leftarrow \varnothing$.
    \end{itemize}

    \item \textbf{Actions in Rounds of the Form $t = 4j + 4$.} For each round $t = 4j + 4$, where $j \in \{0,1,2,\dots, \secparam\}$, every player $P_i$ does the following:
    \begin{itemize}
        \item If the $j$-th group decision is not to share and player $P_i$ receives no shares from others, i.e., $p^j = 0$ and $l_{i^+}^{t-1} = r_{i^-}^{t-1} = \varnothing$; OR
        \item both the $j$-th group decision and $P_i$'s decision are to share and player $P_i$ receives no shares from others, i.e., $p^j = dc^j_i = 1$ and $l_{i^+}^{t-1} = r_{i^-}^{t-1} = \varnothing$; OR
        \item the $j$-th group decision is to share, player $P_i$'s decision is not to share and player $P_i$ receives exactly one share from others, i.e., $p^j = 1, dc^j_i = 0$ and either $(l_{i^+}^{t-1} = \varnothing, r_{i^-}^{t-1} = s_{i^-}^j)$ or $(l_{i^+}^{t-1} = s_{i^+}^j, r_{i^-}^{t-1} = \varnothing)$, then player $P_i$ requests the game to continue by setting $r_{i}^t \leftarrow \varnothing$ and $l_{i}^t \leftarrow \varnothing$.
        \item Otherwise $P_i$ aborts the game by setting $r_{i}^t \leftarrow \bot$ and $l_{i}^t \leftarrow \bot$.
    \end{itemize}

    \item \textbf{Actions in Round $t = 4\secparam + 5$.} Every player $P_i$  aborts the game by setting $r_{i}^t \leftarrow \bot$ and $l_{i}^t \leftarrow \bot$.
\end{itemize} \\\hline
\end{tabularx}
\caption{Strategy $\hat{s}$ for HT Game}\label{strateg: s_hat for ht game}
\end{table}

\clearpage
\section{Computational Nash Equilibrium}\label{appendix:computationalNash}
For readers interested in a closer comparison between our proposed \emph{pseudo-equilibrium} notion and \emph{computational Nash equilibrium}~\cite{halpern2016computational} --- a recent and closely related attempt to bridge the gap between cryptography and game theory, in this section, we revisit the definition of computational Nash equilibrium and the necessary context for using it. 

\Cref{def: hp computational game} states the formal definition of a computable uniform sequence of games (also referred to as a computational game), which is the specific class of games on which computational Nash equilibrium is defined. 

\begin{definition}[Computable Uniform Sequence of Games, Definition 3.1 in~\cite{halpern2016computational}]
\label{def: hp computational game}\hfill

Let $[c] = \{1, 2, \ldots, c\}$ denote the set of player indices. A computable uniform sequence (or computational game) $\cG = \{G_1, G_2, \ldots\}$ of games is a sequence of extensive-form game $G_i$ that satisfies the following conditions:
\begin{itemize}
    \item All the games $G_1, G_2, \ldots$ in the sequence $\cG$ involve the same set of players, i.e., $\{1, \cdots, c\}$.
    
    \item The histories set $H_n$ for the $n$-th game $G_n$ has the following properties:
    \begin{itemize}
        \item Every action available at a non-terminal history is polynomial-size describable. Formally, there exists a polynomial $p$ such that, for all non-terminal histories $h \in H_n^{NT}$, the set of all action $A(h)$ at the history $h$ is at most size $p(n)$: $A(h) \subseteq \bits^{\leq p(n)}$.
        \item It is efficient to check if a sequence of actions (i.e., a history) is valid in $G_n$. Formally, there exists a PPT algorithm $\cA: \bits^* \mapsto \bits$ such that, for any history $h$, $\mathcal{A}(1^n, h) = 1$ if $h \in H_n$, and $\mathcal{A}(1^n, h) = 0$ otherwise.
    \end{itemize}
    
    \item It is efficient to compute which player moves at a given valid non-terminal history. Formally, there exists a PPT algorithm $P: \bits^* \mapsto [c]$ such that, for any history $h \in H^{NT}_n$, it can compute $P(1^n, h) \in [c]$ correctly.

    \item It is efficient to compute the utility for each player. Formally, for every $i \in [c]$, there exists a PPT utility function $u_i: \bits^* \mapsto \Real$ such that, for every terminal history $h \in H^T_n$, it can compute $u_i(1^n, h)$.
\end{itemize}
\end{definition}

\Cref{def: computational Nash} presents the formal definition of computational Nash equilibrium.

\begin{definition}[Computational Nash Equilibrium, Definition 4.1 in~\cite{halpern2016computational}]
\label{def: computational Nash}
Let $\cG = \{G_1, G_2, \ldots\}$ be a computable uniform sequence of games. Then, the polynomial-time strategy profile $\vec{M}= \{M_1, \ldots, M_c\}$ is a \emph{computational Nash equilibrium} of $\cG$ if, for all players $i \in [c]$ and all polynomial-time strategies $M_i'$ in $\cG$ for player $i$, there exists a negligible function $\epsilon$, such that for all $n$,
\begin{align*}
    \sum_{h \in H_n^T} \psi^{G_n}_{\vec{M}}(h) u_i(h) \geq \sum_{h \in H_n^T} \psi^{G_n}_{(M_i', \vec{M}_{-i})}(h) u_i(h) - \epsilon(n),
\end{align*}
where $H_n^T$ denotes the set of terminal histories in the $n$-th game $G_n$, $\psi^{G_n}_{\vec{M}}(\cdot)$ and $\psi^{G_n}_{(M_i', \vec{M}{-i})}(\cdot)$ are the probability distributions over $H_n^T$ induced by playing the strategy profiles $\vec{M}$ and $(M_i', \vec{M}{-i})$, respectively, and $u_i(\cdot)$ is the utility function of player $i$.
\end{definition}

To analyze the computational Nash equilibrium, rather than reasoning directly from the definition --- which is often challenging due to the complexity of the equilibrium's definition and computational games --- the authors propose studying the strategies of the underlying game (or ideal game) that the computational game (or real game) represents. \Cref{def: ideal game represented by the computational game} states necessary conditions that such an underlying game must satisfy, and formally define what it means for a computational game to represent a standard extensive-form game.

\begin{definition}[Represented Underlying Game, Definition 3.3 in~\cite{halpern2016computational}]
\label{def: ideal game represented by the computational game}
Let $\cG = \{G_1, G_2, \ldots\}$ be a computable uniform sequence and $G$ be a fixed extensive-form game. Denote by $f = \{f_1, f_2, \ldots\}$ a sequence of mappings such that each $f_n: H_n \mapsto H$ maps histories ($H_n$) in the $n$-th game $G_n$ to histories ($H$) in the ideal game $G$. Denote by $\cF = \{\cF_1, \cF_2, \ldots\}$ a sequence of mappings such that each $\cF_n$ maps every strategy $\sigma$ in the ideal game $G$ to a corresponding strategy in the $n$-th game $G_n$. We say that a computational game $\cG$ (or real game) $\langle f, F\rangle$-represents the underlying game (or ideal game) $G$ --- or equivalently, that the real game $\cG$ is $\langle f, F\rangle$-corresponding to the ideal game $G$ --- if the quadruple $(\cG, G, f, \cF)$ satisfies the following properties:

\begin{itemize}
    \item[UG1.] $G$ and all the games $G_1, G_2, \ldots$ in the sequence $\cG$ involve the same set of players, i.e., $\{1, \cdots, c\}$.

    \item[UG2.] The computational game $\cG$ should have the same structure of the ideal game $G$. More concretely, there exists a sequence of mappings $f = \{f_1, f_2, \ldots\}$ such that each $f_n: H_n \mapsto H$ maps histories ($H_n$) in the $n$-th game $G_n$ to histories ($H$) in the ideal game $G$, and each $f_n$ is surjective. Furthermore, the surjective mapping $f_n$ has the following properties:
    \begin{itemize}
        \item[(a)] Any history $h \in H_n$ in the computational game $G_n$ and its corresponding history $f_n(h) \in H$ in the ideal game $G$ should contain the same number of actions. Formally, $|h| = |f_n(h)|$.
        \item[(b)] For any history $h \in H_n$ and its corresponding history $f_n(h) \in H$, the same player is assigned to move at both $h$ and $f_n(h)$; and
        \item[(c)] if $h'$ is a subhistory of $h$, then $f_n(h')$ is a subhistory of $f_n(h)$; and
        \item[(d)] if $h$ and $h'$ are in the same information set in $G_n$, then $f_n(h)$ and $f_n(h')$ are in the same information set in $G$.
        \item[(e)] For any distinct histories $h, h' \in H_n$, if $h$ and $h'$ are in the same information set in the game $G_n$, then for any action $a$ such that $h \| a \in H_n$ (i.e., action $a$ is valid after history $h$), the last action in the corresponding mapped histories must be the same; i.e., $\mathrm{LA}(f_n(h \| a)) = \mathrm{LA}(f_n(h' \| a))$.
    \end{itemize}

    \item[UG3.] For any terminal history $h \in H^T_n$ and its corresponding terminal history $f_n(h) \in H^T$, the utility of each player $i \in [c]$ must be the same in both $h$ and $f_n(h)$, i.e., $u_i(h) = u_i(f_n(h))$.

    \item[UG4.] Every strategy in the ideal game $G$, along with any unilateral deviation from it, can be faithfully simulated by a corresponding polynomial-time strategy (and deviation) in the computational game $\mathcal{G}$. More concretely, there exists a sequence of mappings $\cF = \{\cF_1, \cF_2, \ldots\}$ such that each $\cF_n$ maps every strategy $\sigma$ in the ideal game $G$ to a corresponding strategy in the $n$-th game $G_n$. Furthermore, the total mapping $\cF_n$ has the following properties:
    \begin{itemize}
        \item[(a)] The mapping $\cF_n$ applies independently to each player’s strategy in the game $G$. Formally, for every strategy profile $\vec{\sigma} = \{\sigma_1, \cdots, \sigma_c\}$ in $G$, its corresponding strategy profile $\cF_n(\vec{\sigma})$ in the game $G_n$ has the form $\cF_n(\vec{\sigma}) = (\cF_n(\sigma_1), \ldots, \cF_n(\sigma_c))$.
        \item[(b)] For every strategy $\sigma_i$ of player $i \in [c]$ in the ideal game $G$, there exists a PPT strategy $M^{\sigma_i}$ in the computational game $\cG$ that simulates the behavior of $\sigma_i$. More concretely, let $r_n$ and $v_n$ denote the randomness and view of player $i$ when playing the corresponding strategy $\mathcal{F}_n(\sigma_i)$ in the computational game $G_n$. The algorithm $M^{\sigma_i}$ takes $(1^n, v_n, r_n)$ as input and outputs the next action for player $i$. For every strategy $\sigma_i$, there must exist a PPT algorithm $M^{\sigma_i}$ such that its output equals the last action in the history $f_n(\mathcal{F}_n(\sigma_i)(v_n, r_n)) \in H$ --- that is, the corresponding ideal-game history obtained by mapping the computational history (produced by playing $\mathcal{F}_n(\sigma_i)$ with view $v_n$ and randomness $r_n$) back to the ideal game via $f_n$. Formally, 
        \begin{align*}
            \mathrm{LA}(f_n(\mathcal{F}_n(\sigma_i)(v_n, r_n))) = M^{\sigma_i}(1^n, v_n, r_n).
        \end{align*}

        \item[(c)]         
        For all strategy profiles $\vec{\sigma}$ in $G$, all players $i$, and all polynomial-time strategies $M_i'$ for player $i$ in $\mathcal{G}$, there exists a sequence $\{\sigma_1', \sigma_2', \ldots \}$ of strategies for player $i$ in $G$ such that 
        \begin{align*}
            \left\{ \phi_{M_i', \mathcal{F}(\vec{\sigma}_{-i})}^{G_n} \right\}_n \text{ is computationally indistinguishable from } \left\{ \rho_{(\sigma_n', \vec{\sigma}_{-i})}^G \right\}_n,
        \end{align*}
        where $H_n^T$ and $H^T$ denote the sets of terminal histories in the $n$-th computational game $G_n$ and the ideal game $G$, respectively; and the distribution $\phi_{M_i', \mathcal{F}(\vec{\sigma}{-i})}^{G_n}$ is the probability distribution over $H_n^T$ induced by playing the strategy profiles $(M_i', \mathcal{F}(\vec{\sigma}{-i}))$ in $G_n$; and similarly, $\rho{(\sigma_n', \vec{\sigma}{-i})}^G$ is the probability distribution over $H^T$ induced by playing the strategy profiles $(\sigma_n', \vec{\sigma}_{-i})$ in $G$.
    \end{itemize}
\end{itemize}
\end{definition}

Finally, the following theorem shows how a computational Nash equilibrium in the computational game $\mathcal{G}$ can be obtained from an underlying game $G$ and a corresponding strategy mapping $\mathcal{F}$.

\begin{theorem}[Theorem 4.2 in~\cite{halpern2016computational}]
If the computational game $\cG \ \langle f, \cF \rangle$-represents the ideal game $G$ and $\vec{\sigma}$ is an Nash equilibrium in $G$, then $\mathcal{F}(\vec{\sigma})$ is a computational Nash equilibrium of $\cG$.
\end{theorem}

%
%
%
%

\end{document}